\documentclass[11pt,letterpaper]{article}

\usepackage[round]{natbib}
\let\shortcite\citeyearpar
\let\cite\citep

\usepackage[alwaysadjust]{paralist}

\usepackage{ifthen}
\usepackage{latexsym}
\usepackage{amsmath}
\usepackage{amssymb}
\usepackage{amsthm}

\usepackage{thmtools}
\usepackage{thm-restate}
\newif\ifshowrichnessappendix
\newif\ifshowcomplexityappendix
\newif\ifshortenedmaintext
\newif\ifshownotesappendix
\newif\ifdraft
\newif\ifmainpaper
\newif\ifsupplemental

\mainpapertrue
\supplementaltrue
 \shortenedmaintextfalse
\drafttrue
\draftfalse
\showrichnessappendixtrue
\showcomplexityappendixtrue

\ifmainpaper\else
  \ifsupplemental\else
    \errmessage{One of main paper of supplemental material must be generated.}
    \stop
  \fi
\fi   

\ifmainpaper
  \ifsupplemental \else
    \usepackage[1-7]{pagesel}
  \fi
\fi

\ifsupplemental
  \ifmainpaper\else
    \showrichnessappendixtrue
    \showcomplexityappendixtrue
    \usepackage[8-]{pagesel}
  \fi
\fi

\declaretheorem[name=Theorem,parent=section]{theorem}
\declaretheorem[name=Definition,sibling=theorem]{definition}
\declaretheorem[name=Lemma,sibling=theorem]{lemma}

\ifdraft
\newcommand{\hsnote}[1]{\footnote{{\bf {#1}-- HS}}}
\newcommand{\lhnote}[1]{\footnote{{\bf {#1}-- LH}}}
\newcommand{\ehnote}[1]{\footnote{{\bf {#1}-- EH}}}
\else
\newcommand{\hsnote}[1]{}
\newcommand{\ehnote}[1]{}
\newcommand{\lhnote}[1]{}
\fi

\newcommand{\mathtext}[1]{\ensuremath{\mathrm{\text{#1}}}}

\newcommand{\set}[1]{\{#1\}}

\newcommand{\tgsr}{\ensuremath{{\calt\hbox{-}\gsr}}}
\newcommand{\tpsr}{\ensuremath{{\calt\hbox{-}\psr}}}
\newcommand{\psr}{\ensuremath{{\rm PSR}}}
\newcommand{\fpsr}{\ensuremath{{\rm FPSR}}}
\newcommand{\pool}{\ensuremath{{\mathit{Pool}}}}
\newcommand{\psrnormzero}{\ensuremath{{\rm PSR}_{\normzero}}}

\newcommand{\gsr}{\ensuremath{{\rm GSR}}}
\newcommand{\card}[1]{{ \mathopen\parallel {#1} \mathclose\parallel }}
\newcommand{\calf}{\ensuremath{{\cal F}}}
\newcommand{\cale}{\ensuremath{{\cal E}}}
\newcommand{\calt}{\ensuremath{{\cal T}}}
\newcommand{\condition}{\,\mid \:}
\newcommand{\fp}{{\rm FP}}
\newcommand{\fdtime}{{\rm FDTIME}}
\newcommand{\score}[1]{{{\mbox{\it{score}}(#1)}}}
\newcommand{\scoresub}[2]{{{\mbox{\it{score}}_{#1}(#2)}}}
\newcommand{\naturals}{\ensuremath{\mathbb N}}
\newcommand{\integers}{\ensuremath{\mathbb Z}}
\newcommand{\rationals}{\ensuremath{\mathbb Q}}
\newcommand{\nonnegrationals}{\ensuremath{\mathbb Q_{\ge0}}}
\newcommand{\normzero}{\ensuremath{\mathtext{norm-0}}}
\newcommand{\normgcd}{\ensuremath{\mathtext{norm-gcd}}}

\title{A Control Dichotomy for Pure Scoring Rules%
\thanks{Supported in part by NSF grants CCF-\{0915792,\allowbreak{}1101452,\allowbreak{}1101479\} and 
by COST 
Action IC1205.
Work done in part 
while H.~Schnoor visited RIT supported 
by an STSM grant of Cost Action IC1205.
}}

\newcount\hour  \newcount\minutes  \hour=\time  \divide\hour by 60
\minutes=\hour  \multiply\minutes by -60  \advance\minutes by \time
\def\mmmddyyyy{\ifcase\month\or Jan\or Feb\or Mar\or Apr\or May\or Jun\or Jul\or
  Aug\or Sep\or Oct\or Nov\or Dec\fi \space\number\day, \number\year}
\def\hhmm{\ifnum\hour<10 0\fi\number\hour :%
  \ifnum\minutes<10 0\fi\number\minutes}

\author{Edith Hemaspaandra\\
Department of Computer Science\\
Rochester Institute of Technology\\
Rochester, NY 14623, USA \\
www.cs.rit.edu/$\sim$eh
\and Lane A.~Hemaspaandra\\
Department of Computer Science\\
University of Rochester\\
Rochester, NY 14627, USA \\
www.cs.rochester.edu/u/lane%
\and Henning Schnoor \\
Institut f\"ur Informatik\\
Christian-Albrechts-Universit\"at zu Kiel \\
24098 Kiel, Germany \\
henning.schnoor@email.uni-kiel.de
}

\date{April 17, 2014}

\begin{document}
\sloppy

\maketitle

\begin{abstract}
  Scoring systems are an extremely important class of election systems.
  A length-$m$
  (so-called) scoring vector applies only to $m$-candidate elections.  To
  handle general elections, one must use a family of vectors, one per
  length.  The most 
elegant approach to making sure such
  families are ``family-like'' is the recently introduced notion of 
  (polynomial-time uniform) pure scoring 
  rules~\cite{bet-dor:j:possible-winner-dichotomy}, where each scoring vector
  is obtained from its precursor by adding one new coefficient.
  We obtain the
  first
  dichotomy theorem for pure scoring rules for a control
  problem.  In particular, for constructive control by adding voters (CCAV),
  we show that CCAV is solvable in polynomial time for $k$-approval with 
  $k\leq 3$, $k$-veto with $k\leq 2$, every 
  pure scoring rule in which only the two top-rated
  candidates gain nonzero scores, and a particular rule that is a ``hybrid'' of $1$-approval and $1$-veto.
  For all other pure scoring rules, CCAV is NP-complete.
  We also investigate the descriptive richness of different models for
  defining pure scoring rules, proving how more rule-generation time
  gives more rules, proving that rationals give more rules than do the
  natural numbers, and proving that some restrictions previously
  thought to be ``w.l.o.g.''\ in fact do lose generality.
\end{abstract}

\section{Introduction}\label{s:intro}

Elections give rise to a plethora of interesting questions in the social
and political sciences, and have been extensively studied from a
computer-science point of view in the last two decades. We study the
\emph{control problem}, in which the chair of an election (ab)uses her
power to try to 
affect the election outcome.
\hsnote{This already restricts the discussion to constructive control, 
maybe we can use something like
``to skew the election result,'' especially since we describe constructive control
in the next sentence.}%
\lhnote{LATER EDIT: Ok... I've adjusted it to handle the issue raised in 
Henning's hsnote.}%
In this paper we focus on 
\emph{constructive control by adding voters}
(CCAV), i.e., where 
the chair tries to make 
her favorite candidate win by adding voters.
Constructive control by adding voters is an extremely important 
control type, since it occurs in (political) practice very often. 
A standard example is the ``Get-out-the-Vote'' efforts of 
political parties, aimed at (supposed) supporters of those parties. 
However, we should also mention that,
as has been pointed out 
for example in books by Riker~\shortcite{rik:b:art-of-manipulation} 
and Taylor~\shortcite{tay:b:polsci:social-choice-manipulation},
in modern politics issues of candidate introduction or removal 
(CCAC/CCDC) have become highly important;
the case of Ralph Nader in 
two recent American elections vividly supports this point.

The computational complexity of CCAV and other
forms of control was first studied 
by Bartholdi, Tovey, 
and Trick~\shortcite{bar-tov-tri:j:control},
for
plurality and (so-called) Condorcet elections.
\ifshortenedmaintext
We
\else
In this paper, we 
\fi
study the complexity of CCAV for \emph{pure scoring
  rules}, an 
attractive 
class
introduced by Betzler and Dorn \shortcite{bet-dor:j:possible-winner-dichotomy}
that contains many
important voting systems. A scoring rule for an election with $m$
candidates is defined by $m$ coefficients
$\alpha_1\ge\alpha_2\ge\dots\ge\alpha_m$. Each voter ranks the $m$
candidates from her most favorite to her least favorite; a candidate
gains $\alpha_i$ points from being in position $i$ on that 
\ifshortenedmaintext
\else
voter's
\fi
ballot.

Well-known examples of families of scoring rules include the
following.  Borda Count for $m$ candidates uses coefficients $m-1,m-2,\dots,1,0$.
$k$-approval uses coefficients
$\underbrace{1,\dots,1}_k,\underbrace{0,\dots,0}_{m-k}$;  $k$-veto
uses
$\underbrace{1,\dots,1}_{m-k},\underbrace{0,\dots,0}_k$.  Dowdall
voting, used for Nauru's parliament, uses
$1,\frac12,\frac13,\dots,\frac1m$.

The construction of the scoring vector for a specific number of
candidates usually follows a natural pattern, as in the above
examples. This leads to the definition of a ``pure scoring rule.''  We
discuss the 
notion 
of ``purity'' in detail in this paper.  (Basically, it means 
that at each length we insert one 
entry 
into the previous 
length's vector;
\hsnote{strictly speaking, the ``computationally simple way'' 
is enforced by uniformity, not purity}%
\lhnote{LATER EDIT: Excellent point.  I've reedited it.  I
used ``entry'' rather than ``coefficient'' at this one point as 
otherwise ``one (new) coefficient'' might sound as if the SET of VALUES
over all the coefficients can increase in size by at most one, but with
their multiplicities changing.  But let us not clone the word ``entry'' to
other places---the problem is just that right here, ``coefficients'' might
be misread.}%
all our above examples are pure.)

There is a rich literature on computational aspects of scoring rules,
e.g., dichotomy theorems on weighted manipulation
\cite{hem-hem:j:dichotomy}, the possible winner problem
\cite{bet-dor:j:possible-winner-dichotomy,bau-rot:j:possible-winner-dichotomy-final-step},
and bribery \cite{fal-hem-hem:j:bribery}, as well as results about
specific voting systems
\cite{DBLP:conf/ijcai/BetzlerNW11,dav-kat-nar-wal:c:complexity-and-algorithms-for-borda,fal-hem-hem:c:weighted-control}.

In this paper, we provide
\hsnote{is ``perform'' a good word? Maybe ``undertake'' is better. Also,
it's probably not just the first, but also the last complete investigation of the problem, as we 
settle the issue.}%
\lhnote{LATER EDIT: So I've changed ``perform'' to ``provide.''
And I have not said it is the last, as that sounds arrogant to the reader;
but with luck the dichotomy result of the next sentence might make 
that clear... actually, it won't to many readers, but if we change
``the first complete'' to ``a complete'' we'd lose the more important
point that this is a first.  We could say ``we initiate and perform
a complete investigation of'' but then they might complain that 
we had missed some directions, such as typical case or whatever.}%
the first complete investigation of the
complexity of the unweighted CCAV problem for pure scoring rules. We prove
a dichotomy theorem that gives a complete complexity-theoretic
classification of that control problem for pure scoring rules. Our result is
as follows.

It turns out that there are only $4$ types of pure scoring rules for which CCAV is solvable in polynomial time:
\begin{enumerate}
 \item $k$-approval for $k\leq 3$,
 \item $k$-veto for $k\leq 2$,
 \item every pure scoring rule in which only the two top-rated candidates receive a nonzero score,
 \item a particular rule which is 
a ``hybrid'' of $1$-approval and $1$-veto: each voter awards her favorite candidate $1$ point, and her least favorite candidate $-1$ point.
\end{enumerate}

For every pure scoring rule that is not
\lhnote{Henning, your 
removal of my added quotes is problematic I think 
(but please do NOT edit this 
to put that 
or those back, as I have now handled it a different way).  The issue here is 
that a PSR is NOT about the vectors.  IT IS AN ELECTION SYSTEM\@.  Period.
So using equivalent is wrong.  There is no equivalence in play for 
election systems: two are either equal or not equal, period.  
I put quotes around ``equivalent'' earlier to paper over this problem
in the way you were using things there,
but in fact things are worse since in the above you do for item 4 just 
state generator for a rule and then try to identify it with the rule,
which is incorrect.  A generator gives a rule, but different generators
may give the same rule.  The rule is just the election system.  I've 
added ``(generated by)'' above, and after ``one of the above'' have 
added ``election systems'' below, to avoid this confusion and paper over 
the problem.  And yes, equivalent was also misused in the final sentence
of the paragraph, so i rewrote that too.  However, please do NOT 
change things elsewhere; we are so close to the deadline that 
added errors will be hard to remove.  Also, the other uses of 
equivalent are 
about GENERATORS, and there the word is fine.
By the way, you clearly do sometimes have a nice sense of this 
issue, since the place where you just changed 
``it has a $\tgsr$'' into ``it has a generator'' is an example 
of you fixing a place where what I wrote was speaking of a rule (election
system) when it should have been speaking of the generator!
}%
one of the
above election systems, CCAV is NP-complete. The last rule mentioned above is
particularly interesting for two reasons: First, it was the only one
\hsnote{since we repeat ``only one'' in the next sentence, maybe change this
to ``only election system.''}%
\lhnote{LATER EDIT: The change you are suggesting makes the claim overbroad.  
The ``one'' is pulling into this sentence the restriction to PSRs from 
the start of the paragraph, and that is 
needed here.}%
for which the complexity of the possible winner problem was left open
in Betzler and
Dorn~\shortcite{bet-dor:j:possible-winner-dichotomy}. Second, it is
the only one for which polynomial-time solvability depends on the
actual coefficients, and not only on the $<$-order of the values:
While this election system is equal to the rule generated by coefficients
$2,1,\dots,1,0$, the rule 
generated by using coefficients $3,1,\dots,1,0$ leads to
an NP-complete CCAV-problem.

A key point in the proof of any dichotomy theorem is to cover all relevant cases. 
We thus base our dichotomy result on a study of the descriptional
richness of definitions of scoring protocols. We examine how
variations of key parameters as the abovementioned purity
requirement, the complexity allowed to compute the coefficients, and
the universe from which the coefficients may be chosen affect the set
of election systems that can be represented. Interestingly, we
discover that assumptions made previously in the literature, which
were believed to be without loss of generality, in fact restrict the
rules that can be expressed. Taking these results into account, we
introduce a flexible purity condition that is more robust with respect to the
abovementioned variations. We show that the new notion strictly
generalizes pure scoring rules with integer scores and coincides with
pure scoring rules with rational coefficients.

The paper is structured as follows.   
In Section~\ref{sect:prelim}, we give necessary
preliminaries about the class of election systems that we study. In Section~\ref{s:descriptive},
we study the descriptive richness of definitions of positional scoring rules. Section~\ref{sect:dichotomy}
contains our main complexity result. 
All proofs from Section~\ref{s:descriptive} and 
some technical 
proofs from Section~\ref{sect:dichotomy} have been deferred
to the appendix.
\section{Preliminaries}\label{sect:prelim}
As is standard, given an $m$-component (so-called ``scoring'') vector
$
\alpha = (\alpha_{1},\ldots,\alpha_{m})$ such that $\alpha_{1}\geq
\ldots \geq \alpha_{m}$%
, we
define a so-called ($m$-candidate) ``scoring system'' election based
on this by, for each voter (the voters vote by strict, linear orders
over the candidates), giving $\alpha_i$ points to the voter's
$i$th-most-favorite candidate.  Whichever candidate(s) get the highest
number of points, summed over all voters, are the winner(s).
We refer to the $\alpha_i$s as ``coefficients.''

We will use 
the following very 
pure 
definition of pure scoring rules, which removes some of the additional
assumptions 
that have been used in earlier work.
Indeed, earlier work 
asserted that (at least some of) these 
assumptions were not restrictions;
we will look at that issue
anew below.

An election system $\cale$ is a $\tgsr$ 
(generalized scoring rule)
if there
is a function $f$ that on each input $0^m$, $m\geq 1$, outputs an
$m$-component scoring vector $\alpha^m =
(\alpha^m_1,\ldots,\alpha^m_m)$ such that $\alpha^m_1\geq \ldots
\geq \alpha^m_m$ and each $\alpha^i_j$ belongs to $\calt$, and for
each $m$, the winner set under $\cale$ is exactly the winner set given
by the scoring system using the scoring vector $\alpha^m$. 
The notation $0^m$ denotes a string of $m$ ``0''s.
Using this as the 
argument ensures that the generator's computation time is measured
as a function of $m$, since its input length is exactly $m$.  Since the
votes are of size at least $m$ each, this provides the natural, fair approach
to framing FP-uniformity (as we will do below).
We call
$f$ a generator for $\cale$.
While one can consider election systems based on an $f$ that is not computable,
in practice we want there to be an efficient algorithm computing $f$. This is
expressed with different \emph{uniformity} conditions.
If $\cale$ is a $\tgsr$ via some generator $f$ that can be computed in a complexity class
$\calf$, then $\cale$ is an $\calf$-uniform-$\tgsr$ and
$f$ is an $\calf$-uniform $\calt$-generator.
The values of
$\calt$ we will be interested in are the naturals $\naturals$, the 
nonnegative rationals
$\nonnegrationals$%
, 
the rationals $\rationals$,
and the integers $\integers$.  Our most important value of $\calf$ will
be the polynomial-time functions, $\fp$.  When we do not state
``nonuniform'' or some specific uniformity, we always mean
$\fp$-uniform.  
When we put a name 
in boldface, it indicates all the elections that can be generated by 
a generator of the named sort, e.g., \textbf{\boldmath$\fp$-uniform-$\integers$-GSR}
is the class of all polynomial-time-uniform generalized scoring rules
with integer coefficients, a class first defined and discussed by 
Hemaspaandra and Hemaspaandra~\shortcite{hem-hem:j:dichotomy}.

A $\tgsr$ (of whatever uniformity) is a $\tpsr$ 
(pure scoring rule)
if
it has a generator (of the same uniformity) $f$ satisfying the following ``purity'' constraint: For each $m \geq 2$,
there is a component of $\alpha^m$ that when deleted leaves exactly
the vector $\alpha^{m-1}$.  

Throughout the paper we use
\ifshortenedmaintext
that
\else
the following observation, 
which 
notes that
\fi
two scoring vectors, after being ``normalized,'' differ if and only
if they are capturing distinct election systems:
An $m$-position scoring vector $\alpha = (\alpha_1,\ldots,\alpha_m)$ 
(over $\naturals$) is
normalized if $\alpha_m = 0$ and the greatest common divisor of its
nonzero $\alpha_i$'s is 1.  Normalizing a given scoring vector 
(over $\naturals$ or $\integers$) is
easily achievable in polynomial time: 
subtract $\alpha_m$ from each coefficient and then divide each coefficient
by the gcd 
\ifshortenedmaintext
\else
(computed using Euclid's gcd algorithm)
\fi
of the nonzero thus-altered coefficients.  The normalization
of a scoring vector over $\rationals$ 
(resp.,\ $\nonnegrationals$) 
is done by multiplying through by the lcm
of the denominators of the nonzero coefficients, and then viewing that as a 
vector over $\integers$ (resp.,\ $\naturals$) and normalizing it as above.

\begin{restatable}{proposition}{pdiffer}\label{p:differ}
Let $m \geq 1$ and let $\alpha$ and $\alpha'$ be $m$-position
scoring vectors over $\calt\subseteq\rationals$.  Then $\alpha$ and $\alpha'$ have the same winner 
sets on each $m$-candidate election if and only if $\alpha$ and 
$\alpha'$ both have the same normalized version.
\end{restatable}

The if direction basically follows from Observation 2.2
of Hemaspaandra and Hemaspaandra~\shortcite{hem-hem:j:dichotomy}, as noted by
Betzler and 
Dorn~\shortcite{bet-dor:j:possible-winner-dichotomy}.  The 
only if follows by
giving a construction 
that for any two unequal normalized scoring
vectors constructs a vote set on which their winner sets differ.
\ifshortenedmaintext
\else
The construction
works by ``aligning'' the vectors by multiplying each
so that their first coefficients are equal, and then using a padding
construction to ensure that only two candidates are crucial and that
the winner sets can be distinguished by appropriately exploiting the
first position at which the aligned vectors differ. 
\fi

Generators $f_1$ and $f_2$ are \emph{equivalent} if they generate the same election system.
Due to Proposition~\ref{p:differ}, this is the case if and only if, for each length $m$, the
\emph{normalized} scoring vectors generated by $f_1$ and $f_2$ for $m$ candidates are identical.

\section{Descriptive Richness and PSRs}\label{s:descriptive}

We now examine how
amount of time used to generate PSRs, and the universe the PSR's
coefficients are drawn from, affect the family of election rules that
can be obtained.  We also look at whether such a seemingly innocuous
and standard assumption as having the last coefficient always being
zero in fact loses generality; we'll see that it does lose generality,
but in a way that can in part be papered over.  

\subsection{Generation Time Gives Descriptional Richness}
Does more generation time give a richer class of pure scoring rules?
A very tight time hierarchy can be achieved by a legal form
of cheating%
\ifshortenedmaintext
, basically by a setup that in part exploits
the fact that more time can write more 
bits. 
\else
. In particular, consider any (nice) time class that can
for some $m \geq 3$ generate a scoring vector over $\naturals$ of the form
$(\alpha_1, \underbrace{1,\ldots,1}_{m-2},0)$ 
such that $\alpha_1$ is so big that some other
time class 
cannot generate this scoring vector (for example, because
it simply doesn't have enough time to write down enough bits to get a
number as large as $\alpha_1$).  Our vector is normalized, and by
Proposition~\ref{p:differ} this already is enough to 
allow us to argue 
that the
two time classes differ in their winner sets.  
(This proof works for
$\gsr$s.  And it works for $\psr$s if the ``nice''ness of the class allows
it to obey purity 
while employing the above approach---hardly an onerous
requirement.)
However, that 
claim
simply uses the fact that more time can write more 
bits. 
\fi
A 
\ifshortenedmaintext
\else
truly fair and
\fi
far more interesting separation would show 
that one can with more time obtain more pure scoring rules in a way
that does \emph{not} depend on using coefficient lengths that simply cannot
be produced by the weaker time class.  In the dream case, all coefficients
in fact would simply be $0$ or $1$, so there are no long coefficients
in play at all.  We in fact have achieved such a hierarchy theorem.
Full time constructibility
is a standard notion that most natural time 
functions satisfy and,
as is standard,
when we speak of a 
time function $T(m)$ by convention that is shorthand
for $\max(m+1,\lceil T(m) \rceil)$~\cite{hop-ull:b:automata}.
$\fdtime[g(\cdot)]$ denotes the functions computable in the given amount 
of deterministic time.

\begin{restatable}{theorem}{thierarchy}
\label{t:hierarchy}
If $T_2(m)$ is a fully time-constructible function and 
$\limsup\limits_{n\rightarrow\infty} {{T_1(m)\log T_1(m)} \over T_2(m)} = 0$,
then there is an election rule in 
\textbf{\boldmath $\fdtime[T_2(m)]$-uniform-$\{0,1\}$-$\psr$}
that is not in 
\textbf{\boldmath $\fdtime[T_1(m)]$-uniform-$\rationals$-$\gsr$}.
\end{restatable}
The log factor here is not surprising; this is the standard overhead
it takes for a 2-tape Turing machine to simulate a multitape TM\@.  What is
surprising is that no additional factor of $m$ is needed. Why might 
we expect such a factor?
(We caution that the rest of this paragraph is intended mostly
for those having familiarity with the diagonalization techniques 
used to prove hierarchy theorems.)
 In the context of a diagonalization construction
(which is the basic 
 technique used in the proof of Theorem~\ref{t:hierarchy}),
one
might expect to 
(all counting against the overall time $T_2$ limit)
at vector-length $m$ have to ``recreate'' all the
shorter vectors used in earlier diagonalizations to ensure that 
the length-$m$ and length-$(m-1)$ vectors are related in a ``pure'' way.
\hsnote{I rephrased this a little, please check.}%
\lhnote{I've rewritten it a bit more.}%
We however sidestep
the need for that overhead by a ``purity''-inducing trick: For each odd $m$
our scoring vector will be of the form 
$1^{\lfloor m/2\rfloor}0^{\lfloor m/2\rfloor + 1}$, 
and at each even length, we purely extend 
that to whichever of 
$1^{\lfloor m/2\rfloor+1}0^{\lfloor m/2\rfloor + 1}$ or 
$1^{\lfloor m/2\rfloor}0^{\lfloor m/2\rfloor + 2}$ 
diagonalizes against the $T_1$-time machine that is being diagonalized
against (if this is an $m$ when we have time to so diagonalize).  
Briefly, we lurch back to fixed, 
safe way-stations at every second length, and this
removes the need to recompute our own history.

\subsection{Coefficient Richness Gives Descriptional Richness}
The richness and structure of the coefficient set for PSRs
affects how broad a class of election rules can be captured, as 
shown by the following result.
(%
\ifshortenedmaintext
Trivially,
\else
Of course, trivially 
\fi
\textbf{\boldmath $\naturals$-$\psr \subseteq
\nonnegrationals$-$\psr \cap 
\integers$-$\psr \subseteq 
\nonnegrationals$-$\psr \cup
\integers$-$\psr \subseteq 
\rationals$-$\psr$.})

\begin{restatable}{theorem}{talldiffer}
\label{t:all-differ}
\textbf{\boldmath $\nonnegrationals$-$\psr \not\subseteq$ \textbf{nonuniform}-$\integers$-\psr},
\textbf{\boldmath $\integers$-$\psr \not\subseteq$ \textbf{nonuniform}-$\nonnegrationals$-\psr}, and 
\textbf{\boldmath $\rationals$-$\psr \not\subseteq$ \textbf{nonuniform}-$\nonnegrationals$-$\psr \cup {}$ \textbf{nonuniform}-$\integers$-$\psr$}.
\end{restatable}
So for example, 
in pure scoring rules, 
(polynomial-time uniformly)
using
positive rationals 
cannot be simulated by
naturals or integers, even nonuniformly.  
And in pure scoring rules,
(polynomial-time uniformly)
using integers 
cannot
be simulated by naturals or positive rationals, even nonuniformly.
\hsnote{this sentence seems a bit redundant given
the sentence before the actual theorem.}%
\lhnote{LATER EDIT: You are absolutely right.  I've removed the 
more technical sentence before the theorem, and have left the 
more informal one after, but with some of the technical details 
added.}%
One might think 
these claims are impossible, and 
that by normalizing one can go back and forth, but 
it is precisely the purity requirement that is making that sort of 
manipulation impossible---there is a price to purity, and it is showing 
itself here.
\ifshortenedmaintext
\else
(The final part of the theorem does not follow 
automatically from the first two parts plus the trivial observation
before the theorem; just the weaker variant of that part in which
$\cup$ is replaced with $\cap$ follows from those.)
\fi
Note that enlarging
the universe does not necessarily lead to a larger
class of election systems: For example, requiring that coefficients are
odd
natural numbers gives the same set of election systems as merely requiring
them to be natural numbers.

We mention a more flexible and highly attractive 
notion of purity that erases the differences just discussed.
$\calt$-FPSRs (flexible pure scoring rules)%
\ifshortenedmaintext
\else
, of whatever 
uniformity or nonuniformity,
\fi
 will be defined exactly
the same way as $\calt$-PSRs were defined, except the purity condition
is changed to: 
For each $m \geq 2$, there is a component of $\alpha^m$
whose removal gives a scoring vector equivalent to 
$\alpha^{m-1}$. Due to Proposition~\ref{p:differ}, this means that the
relevant scoring vectors have the same normalization.
We call such generators \emph{flexible}.
For 
\emph{this} notion
we have 
\ifshortenedmaintext
\else
for the nonuniform case and the $\fp$-uniform case (and most 
other nice cases), 
\fi
the following equality.

\begin{restatable}{theorem}{tfpsr}
\label{t:fpsr}
\textbf{\boldmath $\naturals$-$\fpsr=\rationals$-$\fpsr = \nonnegrationals$-$\fpsr= \integers$-$\fpsr=\rationals$-$\psr$}.
\end{restatable}

\subsection{Having Smallest Coefficient of Zero Loses Generality---Slightly}

Betzler and Dorn~\shortcite{bet-dor:j:possible-winner-dichotomy} in their
definition of scoring rules require that at each $m$, we have
$\alpha^m_m = 0$ (let us call this condition \normzero), 
and comment that this is not a restriction.  We
note that there are PSRs that cannot be generated by any pure
scoring rule meeting that constraint. What is at issue here is a bit
subtle: At each fixed length, the restriction is innocuous.  But in
the context of families that are bound by the purity constraint, the
restriction loses generality.  On the other
hand, we will also note that each pure scoring rule has 
a
generator 
that is ``close'' to meeting that constraint---it meets it at all
but finitely many lengths.  

Betzler and Dorn~\shortcite{bet-dor:j:possible-winner-dichotomy} also have
a ``gcd is 1'' condition (although the phrasing is not crystal clear
as to whether the gcd constraint applies to the union
\hsnote{actually this probably should be set, because
the union operator needs to be applied to sets, not coefficients, but I'm not sure
whether that change would increase readability}%
\lhnote{LATER EDIT: i see what you mean and you're formally right, but
let us leave it.  if one changes to ``set'' readers may 
think we just mean every coefficient that itself appears at every 
length, which actually would be just one of them, but never mind.
in a full version, we could just write it out in math and so be 
explicit, but that might also be less clear to readers since math 
has to be parsed.}%
of all nonzero
coefficients that occur over all lengths, or whether it must in fact
apply separately at each length; 
the latter (which we call \normgcd) would block the vector $(2,0)$ but the former
would not if the next step were, for example, $(3,2,0)$).  However, 
the gcd issues also can be made to go away ``almost 
everywhere.''
In particular, we have established the following.

\begin{restatable}{theorem}{tnormzero}
\label{t:norm-zero}
There is an $\fp$-uniform
$\psr$ that is not generated even by any nonuniform $\psrnormzero$ generator.  
On the other hand,
every $\fp$-uniform 
\ifshortenedmaintext
(resp.,\ nonuniform)
\else
(respectively, nonuniform)
\fi
$\psr$ is generated by a 
$\fp$-uniform 
\ifshortenedmaintext
(resp.,\ nonuniform)
\else
(respectively, nonuniform)
\fi
$\psr$ generator that 
for all but a finite number of $m$
\lhnote{I've undone your edit, Henning.  Your change misreads 
the quantification regarding norm-0... the natural read of its def is 
that the forall-m is PART of norm-0.  the same applies to normgcd.
so your use of those in the theorem makes no sense, formally speaking.}%
has the property that the last coefficient,
$\alpha_{m,m}$,
is zero and the gcd of the nonzero coefficients in the length-$m$ vector
is one.
\end{restatable}

Does the second 
sentence of the theorem imply that each $\psr$ has all its vectors the same
at each length (except for a finite number of exceptional lengths) as
the vectors of some $\psr$ that satisfies \normzero\ and \normgcd?  The
answer is actually ``no.''  The somewhat subtle issue at play is that
$\psr$s can generate vectors that no generator satisfying \normzero\ and
\normgcd\ can ever generate, such as the family $(3,2,\ldots,2,0)$.
So one should not from our above theorem claim that it follows from
the main dichotomy \emph{theorem} of
Betzler and 
Dorn~\shortcite{bet-dor:j:possible-winner-dichotomy}, as completed by
Baumeister
and Rothe~\shortcite{bau-rot:j:possible-winner-dichotomy-final-step}, 
that we can 
read
off the complexity
(of the possible winner problem)
even in our slightly more flexible case.
However,
our above theorem does---in light of the actual \emph{proof case
  decomposition} used in 
  those papers%
\ifshortenedmaintext
\else
\ (which is based on
issues such as whether one has an unbounded number of positions that
differ and so on)
and some additional argumentation to connect to that
and in particular to note that Betzler and Dorn (and Baumeister and
Rothe) are in effect quietly covering well even those cases that do not
satisfy gcd constraints%
\fi
---connect so well to their work that each of
our cases is settled by 
their proofs.

\section{A Control Dichotomy for PSRs}\label{sect:dichotomy}

We study the following problem for an election system $\cale$: When $R$ is a set of registered voters, is there some subset of the unregistered voters $U$ of size at most $k$ that we can add to the election to ensure that $p$ is the winner? 

\begin{definition}
  Let $\cale$ be an election system. The \emph{constructive control
    problem for $\cale$ by adding voters}, $\cale$-CCAV, is the
  following problem: Given two multisets sets of votes $R$ and $U$, a
  candidate $p$ and a number $k$, is there a set $A\subseteq U$ with
  $\card A\leq k$ such that $p$ is a winner of the election if the
  votes in the multiset $R\cup A$ are evaluated using the system
  $\cale$?
\end{definition}  

We often use a generator, $f$, as a shorthand for the election system
\hsnote{this seems to be the only place
in the paper where we use ``scoring rule family,'' and I'm not sure this
is well-defined. Simply writing ``we identify an election system and its
generator'' is wrong, since there is no unique generator, but maybe we can
write ``identify a generator and the corresponding election system.''}%
\lhnote{LATER EDIT: 
Rewritten to sidestep this issue.... we now say that we often
represent a scoring rule family BY ONE OF its generators.  I hope that this
papers over the issue without causing new problems.}%
\lhnote{LATER EDIT: 
And I just re-rerewote this differently from what I said in 
the above note.}%
(scoring rule family) it generates, 
e.g., 
we write $f$-CCAV. For a generator $f$, we use 
$\alpha^{f,m}=(\alpha^{f,m}_1,\dots,\alpha^{f,m}_m)$ to denote the
scoring vector generated by $f$ for $m$ candidates and its individual
coefficients. 
\ifshortenedmaintext
\else
To simplify presentation, we only consider $\fp$-uniform generators.
However our results continue to hold as long as we can solve the following question in polynomial time: Given $m$, $i$, and $j$ in \emph{unary}, does $\alpha^m_i>\alpha^m_j$ hold, where $f(m)=(\alpha^m_1,\dots,\alpha^m_m)$?
\fi

Our main result is a complexity dichotomy for $f$-CCAV when $f$ is an
$\fp$-uniform 
pure $\rationals$-generator (or, equivalently due to Theorem~\ref{t:fpsr}, 
a $\fp$-uniform flexible $\naturals$-generator). Recall that equivalent 
\hsnote{HS: Moved the ``flexible'' in front of the uniformity condition to be
consistent with the pure $\fp$-uniform directly above.}%
\lhnote{LATER EDIT: Since we often add FP-uniform before PSR since in PSR the pure is 
the "P," it makes more sense for pure to come AFTER the uniformity,
and the same for "flexible."
we don't have to be wildly consistent about this, though, though I 
have swept through quickly to do this type of thing globally.
and actually, now that I think about it, the separating out of Pure 
and Flexible BEFORE the coefficient type clashes with the rest of the 
paper, e.g., we write FP-uniform PSR all over, thought that is about 
rules rather than generators.  Though at some places we DO 
use that type of approach for generators, e.g. in Theorem 2.4.
Actually, things like $\naturals$-generator, which you use above and 
in later theorems, are not even defined in the paper.  However, since 
it is clear enough, let us leave it, and DO NOT GLOBALLY
EDIT REGARDING THIS---we're at the 
point where edits may be doing more harm than good and 
anyway there is no time 
to check/reedit them, plus, the current state is clear enough.}%
generators result in the same election system, hence, due to 
Proposition~\ref{p:differ}, Theorem~\ref{theorem:dichotomy} implies polynomial-time 
results for all generators with the same normalization as one below. 
We state our main result.
 
\begin{theorem}
\label{theorem:dichotomy}
 Let $f$ be a \fp-uniform pure $\rationals$-generator. Then $f$-CCAV is solvable in polynomial time if $f$ is equivalent to one of the following generators:
 \begin{itemize}
  \item $f_1=(1,1,1,0,\dots,0)$ (this generates $3$-approval),
\hsnote{removed the reference to \cale, which is not defined anymore}%
    \item $f_2=(1,\dots,1,0)$ or $f_3=(1,\dots,1,0,0)$ ($1$/$2$-veto),
  \item for some $\alpha\ge\beta$, $f_4=(\alpha,\beta,0,\dots,0)$,
\hsnote{LATE EDIT: superfluous word ``or'' removed}%
  \item $f_5=(2,1,\dots,1,0)$.
 \end{itemize}
 In all other cases, $f$-CCAV is NP-complete.
\end{theorem}

Note that $1$- and $2$-approval are covered by the generator $f_4$.
\hsnote{I added this sentence, since we mention $1$- and $2$-approval in 
the intro and abstract, and readers might be confused if this does not show
up in our main theorem. We could also adjust abstract/intro to only
talk about $3$-approval, and not $1$-, $2$-approval, but I think mentioning
these common systems explicitly is probably better.}%
The remainder of the paper contains the proof of
Theorem~\ref{theorem:dichotomy}: Section~\ref{sect:polynomial time}
contains the algorithms for all polynomial-time solvable cases,
Section~\ref{sect:np complete} contains our hardness results 
(with some proofs deferred to the 
appendix%
),
\hsnote{LATE EDIT: Do we want to reference the
supplemental material here?}%
\lhnote{LATER EDIT: 
Please see the reedit I did, but check that it is ok---what I 
wrote promises 
that every proof related to that section is either in the section or 
in the supplemental material.}%
and the proof
of Theorem~\ref{theorem:dichotomy} follows in Section~\ref{sect:dichotomy
  theorem}.

\subsection{Polynomial-Time Results}\label{sect:polynomial time}

The following result is proven by Lin~\shortcite{lin:thesis:elections}.

\begin{theorem}\label{theorem:3 approval and 2 veto}
 $\cale$-CCAV is solvable in polynomial time if $\cale$ is $k$-approval with $k\leq 3$ or $k$-veto with $k\leq 2$. 
\end{theorem}

Due to Proposition~\ref{p:differ}, Theorem~\ref{theorem:3 approval and
  2 veto} implies that CCAV remains polynomial-time solvable for
``scaled'' versions of $k$-approval with $k\leq 3$ or $k$-veto with
$k\leq 2$, i.e., generators of the form
$(\alpha,\beta,\gamma,\delta,\dots,\delta)$ with
$\beta,\gamma\in\set{\alpha,\delta}$ or
$(\alpha,\dots,\alpha,\beta,\gamma)$ with
$\beta\in\set{\alpha,\gamma}$. We now look at a generalization of
$2$-approval: Voters approve of $2$ candidates, and can distinguish
between their first and second choice.
CCAV for this generalization remains
efficiently solvable.
\hsnote{Late Edit: rewrote this paragraph, it did not look 
nice with the parenthesis.}%
\lhnote{LATER EDIT: Slightly rewrote.}%
In contrast, 
Theorem~\ref{theorem:alpha>gamma>0 np complete} 
shows
that our control problem for the corresponding generalization of $3$-approval is NP-hard.

\begin{theorem}\label{theorem:alpha,beta,zero in ptime}
 Let $\alpha\ge\beta$ be fixed. Then $f$-CCAV is polynomial-time solvable for $f=(\alpha,\beta,0,\dots,0)$.
\end{theorem}

\begin{proof}
  Due to Theorem~\ref{theorem:3 approval and 2 veto}, assume
  $\alpha>\beta$. Let $\ell$ be such that
  $\ell(\alpha-\beta)\ge\beta$. Let an instance with candidates
  $C$, favorite candidate $p$,
\hsnote{LATE EDIT: added the fav candidate $p$, light editing of 
the paragraph to keep it within 6 pages}%
  registered voters $R$, and potential voters $U$ be given; let
  $k$ be the number of voters that can be added. Assume $\ell\leq
  k$, otherwise brute-force.  Let $V_1$ ($V_2$) be the set of voters in
  $U$ that put $p$ in the first (second) spot. Clearly, we add 
  voters only from $V_1\cup V_2$. Assume w.l.o.g.\ that
  two voters who vote identically in the first two positions also rank
  the remaining candidates identically.  \ifshortenedmaintext \else In
  particular, two voters in $V_1$ ($V_2$) are different if and only if
  they vote different candidates in the second (first) place.  \fi We
  use the following facts.
 
 \begin{restatable}{prooffact}{extendingvivotersinptime}
 \label{fact:extending vi voters in ptime}
 For $i\in\set{1,2}$, given a set $S\subseteq V_1\cup V_2$, it can be checked in polynomial time whether $S$ can be extended, by adding at most $k-\card S$ voters from $V_i$ to make $p$ win.
 \end{restatable}
  
 \begin{restatable}{prooffact}{wemostlyaddVOnevoters}
 \label{fact:we mostly add V1 voters}
  To make $p$ win, it is never better to add $\ell$ voters from $V_2$ than adding $\ell$ pairwise different voters from $V_1$.
 \end{restatable}
  
 Due to Fact~\ref{fact:we mostly add V1 voters}, we do not have to
 consider solutions that use $\ell$ or more voters from $V_2$ and
 leave $\ell$ or more distinct voters from $V_1$ unused. So if a
 solution exists, we can find one using fewer than $\ell$ voters from
 $V_2$, or leaving fewer than $\ell$ pairwise different voters from
 $V_1$ unused.  For both cases we will test whether there is a
 corresponding solution.
 
 We start with the first case. Since $\ell$ is constant, we can test
 every subset $S\subseteq V_2$ with $\card{S}<\ell$. For each of these
 $S$, we use Fact~\ref{fact:extending vi voters in ptime} to check in
 polynomial time whether $S$ can be extended to a solution by adding
 voters only from $V_1$.
 
 For the second case, we determine in polynomial time whether there is
 a solution that does not leave $\ell$ pairwise different voters from
 $V_1$ unused as follows: We encode the choice of the unused voters
 from $V_1$ as a function $u\colon
 C\rightarrow\set{0,\dots,\card{V_1}}$ that states, for each candidate
 $c$, the number of unused $V_1$-voters placing $c$ second. Since we
 look for solutions satisfying the second case, we only consider
 functions $u$ for which $u(c)>0$ for at most $\ell-1$ 
 candidates.
  
 \ifshortenedmaintext
 We can now brute-force over all subsets $S\subseteq C$ with $\card{S}\leq \ell-1$ (since $\ell$ is constant) and all functions $u$ as above for which $u(c)>0$ exactly for the candidates in $S$, and apply Fact~\ref{fact:extending vi voters in ptime} to search for the votes to be added.%
 \else
 Thus we can brute-force over all of these possibilities as follows:
 
 \begin{enumerate}
  \item In the outer loop, we test every subset $S\subseteq C$ with $\card{S}\leq \ell-1$. Since $\ell$ is a constant, there are only polynomially many of these sets.
  \item For each such $S$, we test all functions $u$ as above for which $u(c)>0$ holds iff $c\in S$. Such a function can be regarded as a function $u\colon S\rightarrow\set{1,\dots,\card{V_1}}$. There are $\card{V_1}^{\card S}\leq\card{V_1}^{\ell-1}$ many of these functions. Since $\card{V_1}$ is bound by the input size and $\ell$ is a constant, this number is polynomial in the input size.
 \end{enumerate}
 
 So we can polynomially go through all possibilities of potentially unused $V_1$-voters, which is the same as going through all possible sets $S'$ of \emph{used} $V_1$-voters. For each of these sets $S'$, we again use Fact~\ref{fact:extending vi voters in ptime} to check in polynomial time whether $S'$ can be extended to a solution. This completes the proof.
 \fi
 \end{proof}

 Our final
\hsnote{LATE EDIT: exchanged ``second'' with ``final}%
polynomial-time case is the generator
 $(2,1,\dots,1,0)$. Here every voter ``approves'' one candidate and
 ``vetoes'' another. This case is interesting for two reasons. First,
 it is the only case where the algorithm depends on the coefficients
 itself, as opposed to their $>$-order. Namely, for all
 $\alpha>\beta>0$ with $\alpha\neq 2\beta$, $f$-CCAV with
 $f=(\alpha,\beta,\dots,\beta,0)$ is NP-complete
 (Theorem~\ref{theorem:summary finitely many
   coefficients}.\ref{theorem part:alpha1 neq alpha2 alpha1 neq 2
   alpha2}).
 Second, this case was the only one left open in Betzler
 and Dorn's~\shortcite{bet-dor:j:possible-winner-dichotomy}
 possible winner dichotomy;
 the question was eventually settled by
 Baumeister and
 Rothe~\shortcite{bau-rot:j:possible-winner-dichotomy-final-step}%
\ifshortenedmaintext
.
\else
, who proved NP-completeness.
\fi

\begin{theorem}\label{theorem:2 1star 0 in ptime}
 $f$-CCAV is solvable in polynomial time for the generator $f=(2,1,\dots,1,0)$.
\end{theorem}

\begin{proof}
  Let $C$, $R$, and $U$ be the set of candidates, registered voters,
  and unregistered voters, $p$ the preferred candidate, and $k$ the
  number of voters we can add. We add no voter voting $p$ last, and it
  is never better to add a voter voting $p$ second than to add one
  voting $p$ first. So we first add all voters from $U$ that place
  $p$ in the first position%
  \ifshortenedmaintext , using the obvious greedy strategy if there
  are more than $k$.  \else . If there are more than $k$ of these
  voters, we choose the ones to add with the obvious greedy strategy
  that always picks, among all available votes of the form
  $p>\dots>c$, the one where $c$ currently has the highest score.  \fi
  After this preprocessing, all relevant voters in $U$ vote
  $c_1>\dots>c_2$ with $p\notin\set{c_1,c_2}$. To simplify
  presentation, we use Proposition~\ref{p:differ} and consider $f$ as
  the generator $(1,0,\dots,0,-1)$. Then the score of $p$ is
  determined by the votes in $R$.
 
 We reduce the problem to 
min-cost 
\lhnote{This can be done either with max-flow min-cost (which means
find the cheapest cost among all those that achieve max-flow), OR with the 
simpler problem min-cost flow (which means find the cheapest flow 
that provides flow greater than or equal to FOO).  Since we're 
saturating here, either is fine to state.}%
(network) flow, which can be solved in polynomial time. Let $S=\sum_{c\in C - \set{p}}\score c$. We use the following nodes and edges:
 
 \begin{compactitem}
  \item For each $c\in C\ - \set{p}$, there is a node $c$, additionally, there are source and target nodes $s$ and $t$.
  \item There is an edge from candidate $c_1$ to candidate $c_2$ with cost $1$ and with capacity equal to the number of voters in $U$ voting $c_2>\dots>c_1$.
  \item For each candidate-node $c$, there is an edge from $s$ to $c$ with cost $0$ and capacity $\score{c}$ and an edge from $c$ to $t$ with cost $0$ and capacity $\score{p}$.
 \end{compactitem}
 
 Now $p$ can be made winner with at most $k$ additional voters if and
 only if there is a flow from $s$ to $t$ with value $S$ and cost at
 most $k$: Clearly, network flows with cost at most $k$
 correspond to subsets of $U$ with size at most $k$, and using an edge
 $(c_1,c_2)$ $r$ times corresponds to adding $r$ voters voting
 $c_2>\dots>c_1$, since this vote transfers one point from $c_1$ to
 $c_2$. The capacity of the outgoing edges of $s$ ensure that each
 candidate initially gets the correct number of points (since $S$
 points must be distributed), the edges to $t$ ensure that in the end,
 no candidate may have more points than $p$.
\end{proof}

The above results cover all polynomial-time cases of
Theorem~\ref{theorem:dichotomy}.  We now turn to the NP-complete cases.

\subsection{Hardness Results}\label{sect:np complete}

We use the standard NP-complete problem 3DM (3-dimensional matching).
\begin{definition}
3DM is defined as follows:\\[2pt]
 \begin{tabular}{lp{13cm}}
  Input & Pairwise disjoint sets $X$, $Y$, and $Z$ with $\card X=\card Y=\card Z$, and a set $M\subseteq X\times Y\times Z$.\\
  Question & Is there a set $C\subseteq M$ with $\card{C}=\card{X}$ 
that covers $X$, $Y$, and $Z$?
 \end{tabular}
\end{definition}
\noindent 
We say that $C$ covers $X$ 
(resp.,\ $Y$, $Z$) 
if every element from $X$ (resp.,\ $Y$, $Z$) appears in a tuple of $C$. 
\ifshortenedmaintext
\else
Since $X$, $Y$, and $Z$ are pairwise disjoint, in this case every element from $X$ ($Y$, $Z$) appears in the first (second, third) component of a tuple from $C$. 
\fi
Since 
\lhnote{I removed the condition on cardinality of $M$ as that was incorrect
I think---if we really require $\card{M} = \card{X}$, the problem
falls into P as there is only one possible cover, namely all of $M$.}%
$\card{X}=\card{Y}=\card{Z}$, 
a set $C\subseteq M$ with $\card C=\card{X}$ covers $X$ ($Y$, $Z$) if and only if no two tuples from $C$ agree in the first (second, third) component.
A set $C$ 
covering $X$, $Y$, and $Z$ is called a \emph{cover}. 
\lhnote{WARNING TO HENNING: Your definition built into the 
definition of cover the requirement that $\card{C} = \card{X}$.
The rewritten definition and the following definitional paragraph
removes that.  Of course, in reality, all the covers here will be 
exact covers.  But we took it out of the definition.  If you are 
somewhere USING in a proof the word ``cover'' to BUILD IN exactness,
then what we did just broke those proofs.}%

\subsubsection{Constructing Elections}

In our hardness proofs, we often need to set up the registered voters
to ensure specific scores for the candidates. The following lemma
shows that, if there is a ``dummy'' candidate to whom any surplus
points can be ``shifted,'' we can obtain every set of relative scores
that can be expressed as a polynomial-size linear combination of the
coefficients in the scoring vector.

\begin{restatable}{lemma}{coefficientsrealization}
\label{lemma:coefficients realization}
 Given a scoring vector $(\alpha_1,\dots,\alpha_m)$, and for each $c\in\set{1,\dots,m-1}$, numbers $a^c_1,\dots,a^c_m$ in signed unary, and a number $k$ in unary, we can compute, in polynomial time, votes such that the scores of the candidates 
 when evaluating these votes according to the scoring vector $(\alpha_1,\dots,\alpha_m)$
are as follows: There is some $o$ such that for each $c\in\set{1,\dots,m-1}$,
$\score{c}=o+\sum_{i=1}^ma^c_i
\alpha_i$, and $\score{c}>\score{m}+k
\alpha_1$.
\end{restatable}

The value $o$ in Lemma~\ref{lemma:coefficients realization} is the common offset
for all relevant scores. The actual value of $o$ is irrelevant, since the
winner of the election is determined by the relative scores.
The value $k$ is given so that the computed votes
\hsnote{LATE EDIT: exchanged ``construction'' with ``computed votes''}%
ensure that the dummy candidate $m$
cannot win the election with the addition of at most $k$ voters.
\hsnote{Added this discussion
about what $k$ is for, the lemma might 
look unnecessarily complicated otherwise.}%

\subsubsection{``Many'' Different Coefficients}\label{sect:unbounded coefficients}

We now show that the CCAV-problem is NP-complete for generators using
``many'' different coefficients.
Consider any generator $f$ using (at least)
\hsnote{added the ``at least''}%
$7$
different coefficients for some length $m$. Then
with
$\alpha^{f,m}=(\alpha^{f,m}_1,\alpha^{f,m}_2,\alpha^{f,m}_3,\alpha^{f,m}_4,\dots,\alpha^{f,m}_{m-2},\alpha^{f,m}_{m-1},\alpha^{f,m}_m)$
we know that $\alpha^{f,m}_4>\alpha^{f,m}_{m-2}$. This condition in fact suffices for the CCAV problem to be NP-hard;
the result applies to, e.g., Borda, $3$-veto, and
$4$-approval (the latter two use just two different coefficients, but satisfy $\alpha^{f,m}_4>\alpha^{f,m}_{m-2}$ for $m\ge7$).

For $4$-approval or $3$-veto, NP-hardness
\hsnote{removed ``of CCAV'' to save space}%
can be proven by
positioning the elements of $M$ from a 3DM-instance, along with $p$, in the $4$ top positions
of an unregistered
\hsnote{added the ``unregistered''}%
$4$-approval vote or (without $p$) in the last $3$ positions
of an unregistered
\hsnote{again, added ``unregistered''}%
$3$-veto
vote. In 
\hsnote{replaced ``the more general case'' with ``our cases'' to save space}%
our cases,
we can always ``simulate'' one of
these systems: If $\alpha^{f,m}_4>\alpha^{f,m}_{m-2}$, then being
ranked in one of the first $4$ positions is strictly better than being
ranked in one of the last $3$ positions. Roughly speaking, if ``many''
intermediate coefficients are larger than the last $3$, then the last
$3$ are the ``exception,'' and we can use them to
``simulate'' $3$-veto. On the other hand, if ``many'' intermediate
coefficients are smaller than the first $4$, then the first $4$ are
the ``exception'' and we ``simulate'' $4$-approval.\footnote{For
   generators where both cases apply such as
  $f=(2,2,2,2,1,1,1,\dots,1,0,0,0)$, either reduction
  works.} NP-hardness for both $3$-veto and $4$-approval is proved
by Lin~\shortcite{lin:thesis:elections}; however we use a direct
reduction from 3DM in our generalization.

We start with the ``simulation'' of $3$-veto.
\hsnote{Rephrased the following discussion a bit.}%
\lhnote{LATER EDIT: 
I've rewritten it more.  If you yet again rewrite it, please
do not add ``does not only'' back in yet again, as that phrase isn't 
good English.  But please 
DO check that my rewrite below doesn't break the meaning!!  
Note that I have removed the paragraph break between this paragraph
and the following one, as otherwise the final sentence of this 
one just leaves one on a confused cliff.}%
The statement of the following result is a bit unusual.
It indeed gives a reduction for generators
meeting the condition $\alpha^{f,m}_{3k+1}>\alpha^{f,m}_{m-2}$ for all
$m$.  But beyond that 
the function $g$ gives what we call a ``partial'' reduction
from 3DM to $f$-CCAV for $f$'s that meet the condition for some
values of $m$.
In the proof, the size of the 3DM instance
is artificially enlarged to ensure that this ``partial reduction''
meets an analogue counterpart in such a way that for every generator $f$
that satisfies $\alpha^{f,m}_4>\alpha^{f,m}_{m-2}$ for \emph{some} $m$,
we know that for \emph{each} large enough $m$, one of the two reductions
can be applied. 
(Appendix
Section \ref{subsect:4 approval generalization and partial reduction discussion}
has more on this.)
\lhnote{LATER EDIT: Henning, I put forward pointer currently at the end of the 
paragraph.  But if you think it is better to do so---I am not sure
which is better and I lightly lean to where it is now---you could change
the pointer to remove the period inside the parens and then move it to
RIGHT after `what we call a ``partial'' reduction' so it would be 
inside that sentence.}%

\begin{theorem}\label{theorem:3-veto generalization}
 Let $f$ be an \fp-uniform $\rationals$-generator. Then there exists an \fp-computable function $g$ such that 
 \begin{itemize}
  \item $g$ takes as input an instance $I_{\mathtext{3DM}}$ of 3DM and produces an instance $I_{\mathtext{CCAV}}$ of $f$-CCAV with $m=6k$ candidates, where $k=\card{X}=\card{Y}=\card{Z}$.
  \item If $\alpha^{f,m}_{3k+1}>\alpha^{f,m}_{m-2}$, 
then: 
$I_{\mathtext{3DM}}$ is a positive instance of 3DM iff $I_{\mathtext{CCAV}}$ is a positive instance of $f$-CCAV.
 \end{itemize}
\end{theorem}

\begin{proof}
We write $\alpha_i$ for $\alpha^{f,m}_i$. 
W.l.o.g., let $X=\set{s_1,\allowbreak \dots,\allowbreak 
s_k}$,
$Y=\set{s_{k+1}, \allowbreak \dots, \allowbreak s_{2k}}$, and 
$Z=\set{s_{2k+1}, \allowbreak \dots, \allowbreak s_{3k}}$.
We use the following candidates:
 
 \begin{compactitem}
  \item Each $s_i\in\set{s_1,\dots,s_{3k}}$ is a candidate.
  \item $p$ is the preferred candidate.
  \item There are dummy candidates $d_1,\dots,d_{m-3k-1}$. We assume there are at least $3$ dummy candidates, i.e., $k\ge2$.
 \end{compactitem}

\noindent
 \ifshortenedmaintext
With Lemma~\ref{lemma:coefficients realization}, we choose votes in $R$ such that the \emph{relative} scores (the \emph{actual} scores are nonnegative) are as follows:
\hsnote{LATE EDIT: Added the note that the actual scores are nonnegative
to not confuse the reader. Some light editing on the sentence to save space.}%
 \else
We now use Lemma~\ref{lemma:coefficients realization} to construct the set $R$ of registered voters such that the scores of the candidates are as follows. (In the following, we ``normalize'' the scores of all candidates using the score of $p$ as a base. So we pretend that the number $o$ from the application of Lemma~\ref{lemma:coefficients realization} is zero in order to simplify the presentation, clearly the absolute points of all candidates must be positive and are shifted by the actual number $o$ from the lemma.)
 \fi
\ifshortenedmaintext
$\score{p}=0$;
$\score{s_i}=k
\alpha_1-(k-1)
\alpha_{1+i}-\alpha_{m-r(i)}$, where $r(i)=2,1,0$ depending on whether $s_i\in X,Y,Z$, respectively;
$\score{d_i}<- k
\alpha_1$ for $i\in\set{1,\dots,d_{m-3k-1}}$.
\else

 \begin{compactitem}
  \item $\score{p}=0$.
  \item $\score{s_i}=k
\alpha_1-(k-1)
\alpha_{1+i}-\alpha_{m-r(i)}$, where $r(i)=2,1,0$ depending on whether $s_i\in X,Y,Z$, respectively. 
  (So $\alpha_{m-r(i)}$ is exactly the amount of points that $s_i$ gains from a vote that 
  ``vetoes'' $s_i$,
  see below%
)
  \item $\score{d_i}<- k
\alpha_1$ for $i\in\set{1,\dots,d_{m-3k-1}}$.
 \end{compactitem}
\fi

 Let $M\subseteq X\times Y\times Z$ be the set from $I_{\mathtext{3DM}}$. For each $(x,y,z)=(s_h,s_i,s_j)\in M$ (so clearly $h<i<j$), 
we add an available voter to $U$ voting as follows:

$\begin{array}{lll}
 p    >  s_1      >  \dots >  s_{h-1} >  d_1 >  s_{h+1} >  \dots >  s_{i-1}  > \\
 d_2  >  s_{i+1}  >  \dots >  s_{j-1} >  d_3 >  s_{j+1} >  \dots >  s_{3k}   > \\
 d_4  >  \dots   >   d_{m-3k-1}  >  x  >  y  > z.   
   \end{array}$

 We say that such a vote \emph{vetoes} the candidates $x$, $y$, and $z$, and identify elements of $M$ and the corresponding votes.
\hsnote{LATE EDIT: Added the sentence about identification of $M$'s elements and the votes, because we later speak about the votes forming a cover.}%
 
 We show that the reduction is correct.
 \ifshortenedmaintext
 If the 3DM-instance is positive, one can verify that $p$ wins after adding the voters from the cover.
\hsnote{LATE EDIT: Changed ``corresponding to the cover'' to ``from the cover,'' since we now identify these, and we need the change to save space.}%
 \else
 First assume that the instance of 3DM is positive, and let $C\subseteq M$ be the cover with $\card C=k$. We add the voters corresponding to the elements of $C$ in the obvious way and show that $p$ indeed wins the resulting election. 
 
 To see this, it suffices to show that $p$ has at least as many points as each candidate $s_i$, since by construction, the dummy candidates cannot win the election with adding at most $k$ votes. So let $i\in\set{1,\dots,3k}$. The final score for $p$ and $s_i$ are as follows:
 
 \begin{itemize}
  \item $p$ gains $\alpha_1$ points in each of the $k$ additional votes, so $p$ ends up with exactly $k
\alpha_1$ points.
  \item $s_i$ gains $(k-1)
    \alpha_{1+i}$ points from the $(k-1)$ votes corresponding to
    elements $(x,y,z)\in C$ with $s_i\notin\set{x,y,z}$, and
    $\alpha_{m-r(i)}$ points from the single vote vetoing $s_i$. So
    $s_i$ ends up with a final score of
    $k
\alpha_1-(k-1)
\alpha_{1+i}-\alpha_{m-r(i)}+(k-1)
\alpha_{1+i}+\alpha_{m-r(i)}=k
\alpha_1$
    as well.
 \end{itemize}
 
 \fi
 For the converse, let $C\subseteq M$ be a set of at most $k$ votes whose addition lets $p$ win. If this is not a cover,
\hsnote{removed the work ``case'' appearing here}%
 then there is some $s_i$ that is vetoed in none of the added votes. We now compare the points of $p$ and $s_i$.
 \begin{compactitem}
  \item 
$p$ gains $\alpha_1$ points in each of the $\card C$ additional votes, so $p$ ends up with exactly $\card C
\alpha_1$ points.
  \item 
Since $s_i$ is not vetoed in any new vote, $s_i$ gains $\alpha_{1+i}$ points in each added vote and thus ends up with $k
\alpha_1-(k-1)
\alpha_{1+i}-\alpha_{m-r(i)}+\card{C}
\alpha_{1+i}$ points. 
 \end{compactitem}
Since $\card{C}\leq k$, $\alpha_1\ge\alpha_{i+1}$ and $\alpha_{i+1}\ge\alpha_{3k+1}>\alpha_{m-2}\ge\alpha_{m-r(i)}$, it follows that
$ k
\alpha_1-(k-1)
\alpha_{1+i}-\alpha_{m-r(i)}+\card{C}
\alpha_{1+i}  -\card{C}
\alpha_1\\
    =  \underbrace{(k-\card {C})(\alpha_1-\alpha_{1+i})}_{\ge 0}+\underbrace{\alpha_{1+i}-\alpha_{m-r(i)}}_{>0} > 0$.
So $s_i$ beats $p$ if $C$ is not a cover; 
since by assumption adding $C$ makes $p$ win, $C$ must be a cover.
\end{proof}

\ifshortenedmaintext
Theorem~\ref{theorem:3-veto generalization} 
holds analogously, with a slightly more complex argument, when the condition is replaced with 
\hsnote{replaced the original ``with the condition changed to'' in order to
avoid the repetition of ``with''}%
$\alpha^{f,m}_4>\alpha^{f,m}_{m-3k+1}$:
\hsnote{LATE EDIT: added the note about $4$-approval, since this is what 
we discuss earlier at some length.}%
In this case we ``simulate'' $4$-approval.
Both of these ``partial reductions'' also hold for generators satisfying no purity condition at all. 
However, for the ``complete'' reduction
\hsnote{added ``reduction''}%
\lhnote{LATER EDIT: yikes... i had completely misread what the ``complete'' we 
referring too.  Thanks for clarifying that and sorry for my misreading!}%
in Theorem~\ref{theorem:alpha4 >
  alpha(n-2) np completeness}, we need the purity condition to ensure
that the condition $\alpha^{f,m}_4>\alpha^{f,m}_{m-2}$ remains true
for all $m'\ge m$.
The proof of Theorem~\ref{theorem:alpha4 > alpha(n-2) np completeness} then uses that for $m=6k$, the condition $\alpha^{f,m}_4>\alpha^{f,m}_{m-2}$ implies that one of $\alpha^{f,m}_{3k+1}>\alpha^{f,m}_{m-2}$ or $\alpha^{f,m}_4>\alpha^{f,m}_{m-3k+1}$ holds, and then, for each 3DM-instance, picks the appropriate ``partial'' reduction.
\hsnote{rephrased this discussion as well.}%
\else
In a similar way, we can prove an analogous result for all scoring rules that ``can implement'' $4$-approval in the sense that being voted in one of the first $4$ positions is strictly better than being voted in most ``later'' positions. The proof of the following result is very similar to the proof of Theorem~\ref{theorem:3-veto generalization}, except that an additional argument is needed to ensure that the favorite candidate cannot be made a winner with less than $k$ additional voters.

\begin{restatable}{theorem}{fourapprovalgeneralizationinformal}
\label{theorem:4-approval generalization informal}
 Theorem~\ref{theorem:3-veto generalization} also holds when the condition $\alpha^{f,m}_{k+1}>\alpha^{f,m}_{m-2}$ is replaced with $\alpha^{f,m}_4>\alpha^{f,m}_{m-3k+1}$.
\end{restatable}

As mentioned above, we now put the two reductions above together to obtain the NP-hardness result of this section, i.e., to prove that $f$-CCAV is NP-complete as soon as there is a number $m$ where the coefficients of $f$ satisfy $\alpha^{f,m}_4>\alpha^{f,m}_{n-2}$. If this condition is true, then we know that one of the inequalities $\alpha^{f,m}_4\ge\alpha^{f,m}_5\ge\dots\ge\alpha^{f,m}_{m-3}\ge\alpha^{f,m}_{m-2}$ is in fact strict. Depending on the position of this strict inequality, we choose which reduction to apply: If the strict inequality appears ``close'' to the first candidate, then the first ``few'' positions are strictly better than ``most,'' and the system can ``simulate'' $k$-approval for some $k\ge 4$. On the other hand, if the strict inequality appears ``close'' to the last candidate, then the last ``few'' positions are worse than 
``most,'' and we can similarly ``simulate'' $k$-veto for some $k\ge3$.
\fi

\begin{restatable}{theorem}{alphafourbiggeralphanminustwonpcomplete}
\label{theorem:alpha4 > alpha(n-2) np completeness}
 $f$-CCAV is NP-complete for every \fp-uniform pure $\rationals$-generator $f$ with $\alpha^{f,m}_4>\alpha^{f,m}_{m-2}$ for some $m$.
\end{restatable}

\subsubsection{``Few'' Different Coefficients}

We now study pure generators $f$ \emph{not} covered by
Theorem~\ref{theorem:alpha4 > alpha(n-2) np completeness}, i.e., where
$\alpha^{f,m}_4\leq\alpha^{f,m}_{m-2}$ for all $m$. Then for $m\ge6$,
$\alpha^{f,m}$ is of the form
$(\alpha^{f,m}_1,\allowbreak 
\alpha^{f,m}_2,\allowbreak \alpha^{f,m}_3,\allowbreak 
\alpha^{f,m}_4,\allowbreak \dots,\allowbreak \alpha^{f,m}_4,\allowbreak 
\alpha^{f,m}_5,\allowbreak \alpha^{f,m}_6)$. The
reductions above cannot work in this case, since there are no $3$
positions ``worse than most'' and no $4$ positions ``better than
most.''
\hsnote{LATE EDIT: removed ``so we cannot ``simulate'' $3$-veto or $4$-approval.``}%

Due to Theorem~\ref{t:fpsr}, we can regard $f$ equivalently as
flexible $\naturals$-generators or as pure
$\rationals$-generators. For the latter representation, purity
requires 
that all coefficients from $\alpha^{f,m}$ also appear
in $\alpha^{f,m+1}$. So the above numbers
$\alpha_1,\dots,\alpha_6$ do not depend on $m$. We can use a fixed
affine transformation for these finitely many coefficients and, using
Proposition~\ref{p:differ}, rewrite all coefficients as natural
numbers.
\hsnote{I'm not sure this paragraph is all that interesting anymore.
We reference this discussion in the proof of the dichotomy theorem,
so if we want to remove this here, that proof needs to be adjusted.
I'm not sure we want to make such a change at this point, but if we get
really desperate for space, we might reconsider.}%
\lhnote{LATER EDIT: I haven't edited or changed this at all.  It does fit 
into the 6 pages currently, and once we are beyond the AAAI submission,
we won't have space issues (and will have lots of time), e.g., we can
by extra AAAI pages and journals have more space still, so we at that
point can handle this whatever way you choose to without worrying about
space and just weighing what is overall best in terms of what 
discussions to present.}%

Our next hardness result concerns a generalization of $3$-approval.
Recall from Theorem~\ref{theorem:3 approval and 2 veto} that CCAV for
$3$-approval itself, i.e., the generator
$(\alpha,\alpha,\alpha,0,\dots,0)$, is solvable in polynomial time. In
Theorem~\ref{theorem:alpha,beta,zero in ptime}, we proved a
generalization of $2$-approval to still give a polynomial-time
solvable CCAV-problem. We now show that the analogous generalization
of $3$-approval leads to NP-completeness.

\begin{theorem}\label{theorem:alpha>gamma>0 np complete}
Let $\alpha \geq \beta \geq \gamma > 0$ and $\alpha \neq \gamma$.
Let $f$ be the generator giving $(\alpha,\beta,\gamma,0,\dots,0)$.
\hsnote{added the ``for some $\beta$'',
though it is not strictly necessary, it might confuse the reader
if $\beta$ comes out of nowhere.
One could also add ``with $\gamma\leq\beta\leq\alpha$,'' but
I think it's fine as it is, and (without this footnote) the theorem
only needs $2$ lines this way. (The result in fact would hold for
any $\beta\neq 0$, but of course the construction in the proof would
need to ``sort'' the vector or the votes.)}%
\lhnote{LATER EDIT: 
One should avoid putting 
quantifiers after what they modify---it 
is asking for trouble.
That is especially so when you are putting an existential quantifier 
after its own scope, when there is an implicit universal quantifier;
that is PRECISELY the case here.  The problem is, there are two 
conflicting readings as to quantifier order, 
and since there is no purity required here,
one of them makes f-CCAV undecidable for some cases. I've rewritten 
using one of the other options you mentioned.}%
Then $f$-CCAV is NP-complete.
\end{theorem}

\begin{proof}
  Let $M$ be the set from an instance of 3DM with $\card{M}=n$, and let
  $k=\card{X}=\card{Y}=\card{Z}$ (recall
$X$, $Y$, and $Z$ must be pairwise disjoint).
  We use the candidates $X\cup Y\cup
  Z\cup\set{p}\cup\set{S_i,S_i'\ \vert\ S_i\in M}$ and a dummy
  candidate $d$ to be able to apply Lemma~\ref{lemma:coefficients
    realization}. We use the lemma to set up the registered votes such
  that the resulting relative scores are as follows:
  $\score{p}=\alpha+2\gamma$, $\score{c}=(n+2k)\beta+2\gamma$ for all
  $c\in X\cup Y\cup Z$, $\score{S_i}=(n+2k)\beta+\min(\alpha,2\gamma)$,
  and $\score{S_i'}=(n+2k)\beta+\alpha+\gamma$ for each $S_i\in
  M$. Further, $\score{d}<-(n+2k)
\alpha_1$.
For each $S_i=(x,y,z)$, we introduce four unregistered voters voting as follows:
\ifshortenedmaintext
$x>p>\allowbreak 
S_i>\allowbreak \dots$; $y>p>\allowbreak S_i>\allowbreak \dots$;
$z>p>\allowbreak S_i'>\allowbreak \dots$; $S_i>p>\allowbreak S_i'>\allowbreak \dots~$.
\else
\begin{description}
\item[]
 $x>p>S_i>\dots$~.
\item[]
 $y>p>S_i>\dots$~.
\item[]
 $z>p>S_i'>\dots$~.
\item[]
 $S_i>p>S_i'>\dots~$.
\end{description}
\fi
\ifshortenedmaintext
One can verify that $p$ can win the election by adding at most $n+2k$ voters iff there is a cover $C$ with $\card{C}\leq k$.
\hsnote{Changed the $=k$ to $\leq k$, even though the original wording is also correct.}%
\else

We show that $p$ can be made a winner of the election by adding at most $n+2k$ voters if and only if the 3DM-instance is positive, i.e., there is a set $C\subseteq M$ with $\card{C}=k$ and for $S_i\neq S_j\in C$, $S_i$ and $S_j$ differ in all three components.

First assume that there is such a cover. In this case, $p$ can be made a winner of the election by adding the following voters: For each $S_i=(x,y,z)\in C$, we add the votes $x>p>S_i$, $y>p>S_i$, and $z>p>S_i'$. For each $S_i=(x,y,z)\notin I$, we add the vote $S_i>p>S_i'$. Note that this adds exactly $3
 k+(n-k)=n+2k$ votes. Adding these votes results in the following scores:

\begin{itemize}
 \item $p$ gains $\beta$ points in each added vote, so $p$ gains $(n+2k)\beta$ points and $p$'s final score is $\alpha+2\gamma+(n+2k)\beta$,
 \item each candidate in $X\cup Y\cup Z$ gains $\alpha$ points, leading to a final score of $\alpha+\underbrace{(n+2k)\beta+2\gamma}_{\mathtext{previous score}}$ as well,
 \item for each $S_i\in C$, we have that 
 $\score{S_i}=\underbrace{(n+2k)\beta+\min(\alpha,2\gamma)}_{\mathtext{previous score}}+2\gamma
  \leq
  (n+2k)\beta+\alpha+2\gamma$, which again is the score of $p$.
  \item for each $S_i\notin C$, we have $\score{S_i}=(n+2k)\beta+\min(\alpha,2\gamma)+\alpha\leq 
    (n+2k)\beta+2\gamma+\alpha$, equal to the score of $p$.
  \item for each $S_i'$ (independent of whether $S_i'\in C$ or $S_i'\notin C$), we have 
 $\score{S_i'}=\underbrace{(n+2k)\beta+\alpha+\gamma}_{\mathtext{previous score}}+\gamma=(n+2k)\beta+\alpha+2\gamma$, again this is the score of $p$.

\end{itemize}

Thus all candidates tie and so in particular, $p$ is a winner of the election.

For the converse, assume that $p$ can be made a winner by adding at most $n+2k$ voters. Since each $S_i'$ initially beats $p$, at least one vote is added. Thus there is a candidate $c\in X\cup Y\cup Z$ with $\score{c}\ge (n+2k)\beta+2\gamma+\alpha$, or some $S_i'$ with $\score{S_i'}\ge(n+2k)\beta+\alpha+2\gamma$. In both cases, we need to add at least $n+2k$ voters to ensure that $p$ has at least $\alpha+2\gamma+(n+2k)\beta$ points as well.

Since $n+2k$ votes are added, and each of these votes gives points to $p$ and $2$ other candidates, there are $2n+4k$ positions awarding points in the added votes that are filled with (not necessarily different) candidates other than $p$. Each of the $3k$ candidates from $X\cup Y\cup Z$ can only gain $\alpha$ points without beating $p$ in the election, so each of these can fill at most one of these $2n+4k$ positions. So at least $2n+k$ positions must be filled by (again, not necessarily different) candidates from $\set{S_i,S_i'\ \vert\ 1\leq i\leq n}$. Each $S_i'$ can appear at most once in the third position without beating $p$. Since there are $n$ candidates of the form $S_i'$, it follows that there must be $n+k$ occurrences of candidates $S_i$ in the first three positions of the added votes. Since no $S_i$ can gain $\alpha+\gamma$ points without beating $p$,\footnote{To see this, we compute the difference between the score of $S_i$ after gaining $\alpha+\gamma$ points and that of $p$ after gaining $(n+2k)\beta$ points. This value is $(n+2k)\beta+\min(\alpha,2\gamma)+\alpha+\gamma-\alpha-\gamma-(n+2k)\beta=\min(\alpha,2\gamma)-\gamma$. Since $\alpha>\gamma$ and $\gamma>0$, this value is strictly positive, so $S_i$ indeed beats $p$ if $S_i$ gains $\alpha+\gamma$ points.} each $S_i$ can either appear in a vote $S_i>p>S_i'$, or in up to two votes of the form $c>p>S_i$ with $c\in X\cup Y$. ($S_i=(x,y,z)$ cannot appear in three of these, since then one of $x$ and $y$ would gain too many points.) So the only way to fill $n+k$ positions with candidates of the form $S_i$ is having $2k$ occurrences of $S_i$ in the third place, and $n-k$ occurrences of $S_i$ in the first place. In order to fill all positions, each $S_i'$ has to appear once in the final position, and due to the above, $n-k$ of these occurrences are in a vote of the form $S_i>p>S_i'$. Thus there are $k$ votes of the form $z>p>S_i'$. It follows that there are $3k$ votes added that vote a candidate from $X\cup Y\cup Z$ in the first position, and $n-k$ voters are added that vote some $S_i$ first. Since no $S_i$ may appear both in first and in last position, and each $S_i'$ may appear only once, and each $x_i$, $y_i$, and $z_i$ may gain only $\alpha$ points, it follows that the added votes correspond to a cover.
\fi
\end{proof}

We also have proved the following cases NP-complete.

\begin{restatable}{theorem}{summaryfinatelymanycoefficients}
\label{theorem:summary finitely many coefficients}
 The problem $f$-CCAV is NP-complete
\lhnote{I 
changed NP-hard to NP-complete to match the sentence that came before it.
The list of 4 cases is formally I think a bit of a worry, since the theorem
doesn't itself state purity or keep the value of each coefficient from changing 
over the different $m$s, yet that is what is meant as without those 
it all can be undecidable.  So for example part 1 if we had space, 
which we don't, should be expressed as:
``For some alpha1, alpha2, alpha3, alpha4, and alpha5 
satisfying $\alpha_2 > \alpha_4 > 0$, 
$f=(\alpha_1,\alpha_2,\alpha_3,\alpha_4,\dots,\alpha_4,\alpha_5,0)$.
This makes it explicit that one can't play games as to changing the 
alphas as m changes.  But we don't have space for that here.}%
if $f$ is one of the following pure generators:
\hsnote{I added ``pure'' to address the issue at least briefly}%
 \begin{enumerate}
  \item\label{theorem part:alpha2>alpha4 np complete} $f=(\alpha_1,\alpha_2,\alpha_3,\alpha_4,\dots,\alpha_4,\alpha_5,0)$ with $\alpha_2>\alpha_4>0$.
  \item\label{theorem part:alpha1 neq alpha2 alpha1 neq 2 alpha2} $f=(\alpha_1,\alpha_2,\dots,\alpha_2,0)$ with $\alpha_1\notin\set{\alpha_2,2\alpha_2}$, $\alpha_2>0$.
  \item\label{theorem part:alpha1>alpha2>alpha5} $f=(\alpha_1,\alpha_2,\dots,\alpha_2,\alpha_5,0)$ with $\alpha_1>\alpha_2>\alpha_5$.
  \item\label{theorem part:alpha1=alpha2>alpha5>0} $f=(\alpha_1,\dots,\alpha_1,\alpha_5,0)$ with $\alpha_1>\alpha_5>0$.
 \end{enumerate}
\end{restatable}

\subsection{Proof of Dichotomy Theorem}\label{sect:dichotomy theorem}

We now use the individual results from Sections~\ref{sect:polynomial time} and~\ref{sect:np complete} to prove our main dichotomy result, Theorem~\ref{theorem:dichotomy}:

\begin{proof}
  The polynomial cases follow from Theorems~\ref{theorem:3 approval
    and 2 veto}, \ref{theorem:alpha,beta,zero in ptime},
  and~\ref{theorem:2 1star 0 in ptime}, we prove hardness. If
  $\alpha^{f,m}_4>\alpha^{f,m}_{m-2}$ for some $m$, hardness follows
  from Theorem~\ref{theorem:alpha4 > alpha(n-2) np
    completeness}. So assume
  $\alpha^{f,m}_4=\dots=\alpha^{f,m}_{m-2}$ for all $m\ge 6$. As
  argued in the discussion after Theorem~\ref{theorem:alpha4 > alpha(n-2) np completeness},
\hsnote{replaced the ``earlier'' with a concrete reference to the place 
  where the discussion happens}%
  we assume
  $\alpha^{m,f}=(\alpha_1,\alpha_2,\alpha_3,\alpha_4,\dots,\alpha_4,\alpha_5,\alpha_6)$
  for each $m\ge6$. Due to Proposition~\ref{p:differ}, we can assume
  $\alpha_6=0$. We reduce the number of relevant coefficients from $5$
  to $3$:
 
 \begin{compactitem}
  \item If $\alpha_4=0$, then, since $f$ does not generate $3$-approval and is not equivalent to $f_4$, $\alpha_1>\alpha_3>0$. Hardness follows from Theorem~\ref{theorem:alpha>gamma>0 np complete}.
\hsnote{Changed the sentence to remove a reference to the election system $\cale$, which
  we do not talk about anymore in the theorem statement, which is now phrased
  about generators alone.}%
  \item  If $\alpha_2>\alpha_4>0$, hardness follows from Theorem~\ref{theorem:summary finitely many coefficients}.\ref{theorem part:alpha2>alpha4 np complete}. 
 \end{compactitem}
 
\smallskip

\noindent
So assume $\alpha_2=\alpha_3=\alpha_4>0$, i.e., $f$ is of the form $(\alpha_1,\alpha_2,\dots,\alpha_2,\alpha_5,0)$. We make a further case distinction:
 
 \begin{compactitem}
 \item If $\alpha_2=\alpha_5$, then since $f$ does not generate
\hsnote{again changed the wording to remove the reference to $\cale$}%
   $1$-veto, we
   know that $\alpha_1\neq\alpha_5=\alpha_2$. Since $f$ is not
   equivalent to
\hsnote{again removed reference to $\cale$}%
   $(2,1,\dots,1,0)$, we know that $\alpha_1\neq
   2\alpha_2$. Thus NP-hardness follows from
   Theorem~\ref{theorem:summary finitely many
     coefficients}.\ref{theorem part:alpha1 neq alpha2 alpha1 neq 2
     alpha2}.
  \item If $\alpha_2>\alpha_5$, then depending on whether $\alpha_1>\alpha_2>\alpha_5$ or $\alpha_1=\alpha_2>\alpha_5$, hardness follows from Theorem~\ref{theorem:summary finitely many coefficients}.\ref{theorem part:alpha1>alpha2>alpha5} or Theorem~\ref{theorem:summary finitely many coefficients}.\ref{theorem part:alpha1=alpha2>alpha5>0} (note that in the latter case, we know that $\alpha_5\neq 0$, since $f$ does not generate $2$-veto).
  \qedhere
 \end{compactitem}
\end{proof}

\paragraph*{Acknowledgments} 
We thank the AAAI 2014 reviewers for 
helpful comments and suggestions.

\appendix

\section*{Appendix}

\medskip

The 
appendix
is structured as follows:

\begin{itemize}
 \ifshowrichnessappendix
   \item In Section~\ref{appendix:richness}, we provide additional discussion and full proofs for the results on descriptional richness and pure scoring rules contained in Sections~\ref{sect:prelim} 
and~\ref{s:descriptive} of this paper.
 \fi
 \ifshowcomplexityappendix
   \item Section~\ref{appendix:dichotomy} contains the proofs that were omitted from Section~\ref{sect:dichotomy}, i.e., the dichotomy result, and also
     provides additional discussion of ``partial reductions.''
 \fi
\end{itemize}

\ifshowrichnessappendix

\section{Omitted Proofs and Discussion from Sections~\ref{sect:prelim} and~\ref{s:descriptive}}\label{appendix:richness}

\subsection{Proof of Proposition~\ref{p:differ}}

\pdiffer*

\begin{proof}
The ``if'' direction follows from 
Observation 2.2
of~Hemaspaandra and Hemaspaandra~\shortcite{hem-hem:j:dichotomy}, as noted in
Betzler and Dorn~\shortcite{bet-dor:j:possible-winner-dichotomy}.  

Let us prove the ``only if'' direction.  
This result seems so natural and important that it feels as if it should 
be a folk theorem, although we don't know of it as such; 
but in any case, since the result is crucial
to this paper, we include a construction that clearly and 
explicitly establishes this claim.
Let $D$ and $D'$ be arbitrary, fixed
length-$m$ scoring vectors over $N$
such that their normalizations,
$A$ and $A'$, differ.  We will construct an $m$-candidate 
election in which the winner sets under $A$ and $A'$ 
differ.  

Since by the ``if'' direction of the present proposition $A$ and $D$
always yield the same winner set, and also by the ``if'' direction of the
present proposition $A'$ and $D'$ always yield the same winner set,
we may conclude that our constructed election has different winner
sets under $D$ and $D'$.

Let the components of the normalized scoring vector $A$ be 
$(\alpha_1, \alpha_2, \ldots, 0)$, and 
let the components of the normalized scoring vector 
$A'$ be 
$(\alpha'_1, \alpha'_2, \ldots, 0)$.

If $m=1$, $A \neq A'$ is impossible.
If $m=2$, $A \neq A'$ exactly if one of them is $(0,0)$
and one is $(1,0)$, and these easily can be seen to give different
winners on some inputs.
Thus we from now on in this proof assume that $m \geq 3$.

A scoring vector is trivial if all its coefficients are the same, e.g.,
$(0,0,0)$.
If $A$ and $A'$ are both trivial, then $A=A'$, and if exactly one is 
trivial then building an input separating them is easy.  Thus we from
now on in this proof assume that both $A$ and $A'$ are nontrivial.

Recall that both $A$ and $A'$ are normalized.  We will in polynomial
time ``align'' them, i.e., we will scale them so that the first 
components of $A$ and $A'$ become identical.  In particular, 
to align them we will multiply $A$ by $\alpha'_1$ and we will
multiply $A'$ by $\alpha_1$.  So the first coefficient of each 
will now be $\alpha_1 \alpha'_1$. Let us rename the thus-scaled 
vectors (each equivalent to its original vector as to what winner 
sets it gives, due to 
Observation 2.2
of Hemaspaandra and Hemaspaandra~\shortcite{hem-hem:j:dichotomy}) as 
$B = (\beta_1, \beta_2, \ldots, \beta_{m-1}, 0)$ and 
$B' = (\beta'_1, \beta'_2, \ldots, \beta'_{m-1}, 0)$.

Let $\gamma$ be the least $i$ such that $\beta_i \neq \beta'_i$.
W.l.o.g.,\ assume $\beta_i > \beta'_i$.
$\gamma = 1$ is impossible since $\beta_1 = \beta'_1 = \alpha_1\alpha'_1$.
So $2 \leq i \leq m-1$.

We will now specify an election in which $B$ and $B'$ have different 
winner sets.

Our candidates will be $a$, $b$, and ``dummy'' candidates 
$d_1,\ldots,d_{m-2}$.  We will 
ensure that only $a$ and $b$ are serious contenders for winning.

Let 
$s_1$ be a shorthand for ``$d_1 > d_2 > \cdots > d_{m-2}$,''
let 
$s_2$ be a shorthand for ``$d_2 > d_3 > \cdots > d_{m-2} > d_{1}$,''
and so on, up to 
$s_{m-2}$ being a shorthand for ``$d_{m-2} > d_1 > \cdots > d_{m-3}$.''

Our vote set is as follows:
\begin{enumerate}
\item For each $i$, $1\leq i \leq m-2$, we will have 
$H$ votes  ``$a>b>s_i$''
and we will have $H$ votes
``$b>a>s_i$.''
$H$'s exact value will be specified later in the proof, but will be chosen to
be so large that $a$ and $b$ are the only serious contenders.
\item $\beta_\gamma$ votes of the form ``$a$ is the top choice and 
$b$ is the last choice 
(and the other choices will be irrelevant to 
our proofs, but for specificity let us say they are filled in 
in lexicographical order).''
\item $\beta_1$ votes of the form ``$b$ is the $\gamma$th choice and $a$
is the last choice
(and the other choices will be irrelevant to 
our proofs, but for specificity let us say they are filled in 
in lexicographical order).''
\end{enumerate}
That ends our specification of the votes.

Let us tally up the points that each candidate gets under $B$ and 
under $B'$.  Clearly, and keeping in mind that $\beta_1 = \beta'_1$, 
we have:
$\scoresub{B}{a} = (m-2)H(\beta_1+\beta_2)+\beta_1\beta_\gamma+0$,
$\scoresub{B'}{a} = (m-2)H(\beta'_1+\beta'_2)+\beta_\gamma\beta'_1+0$,
$\scoresub{B}{b} = (m-2)H(\beta_1+\beta_2)+\beta_1\beta_\gamma+0$,
$\scoresub{B'}{b} = (m-2)H(\beta'_1+\beta'_2)+\beta_1\beta'_\gamma+0$,
$\scoresub{B}{d_i} \leq (m-3)(2\beta_3) + \beta_\gamma\beta_2+(\beta_1)^2$,
and 
$\scoresub{B'}{d_i} \leq (m-3)(2\beta'_3) + \beta_\gamma\beta'_2+\beta_1\beta'_1$.
We now set the value of $H$, namely to be $H=2\beta_1$.
Keeping in mind that $\beta_1=\beta'_1 \geq 1$, it is easy to see that 
this choice of $H$
ensures that for each $i$, $1\leq i \leq m-2$,
$\min(\scoresub{B}{a},\scoresub{B}{b}) > \scoresub{B}{d_i}$ and
$\min(\scoresub{B'}{a},\scoresub{B'}{b}) > \scoresub{B'}{d_i}$.  So we
have ensured that, under $B$ and under $B'$, $a$ and $b$ have more
points than any $d_i$.

Under $B'$, note that $a$ is the one and only winner (recall
$\beta_1 = \beta'_1$ and $\beta_\gamma > \beta'_\gamma$).  But 
our votes ensure that 
under $B$, $a$ and $b$ tie as the (sole) winners.

So the votes we gave show that $B$ and $B'$ (and thus $D$ and $D'$)
have different winner sets on some example, namely, the above example.%
\end{proof}

\subsection{Proof of Theorem~\ref{t:hierarchy}}

\thierarchy*

\begin{proof}
  Rather than proving 
  Theorem~\ref{t:hierarchy},
  we will prove a slightly weaker result, ``Theorem~X,''
  and then will explain how to modify that proof to 
  establish 
  Theorem~\ref{t:hierarchy}.

\begin{quote}
{\bf Theorem X}\quad
There is an integer constant $k>0$ such that 
if $T_2(m)$ is a fully time-constructible function, and
$\limsup\limits_{n\rightarrow\infty} {{T_1(m)\log T_1(m)} \over T_2(m)} = 0$,
and 
$(\forall m)[T_2(m) \geq k(m+1)]$,
then there is an election rule in 
\textbf{\boldmath $\fdtime[T_2(m)]$-uniform-$\{0,1\}$-$\psr$}
that is not in 
\textbf{\boldmath $\fdtime[T_1(m)]$-uniform-$Z$-$\gsr$}.
\end{quote}

We will prove Theorem~X by describing how to appropriately adjust the
classic 
presentation by Hopcroft and Ullman~\shortcite{hop-ull:b:automata}, 
henceforward ``HU,'' 
of
the proof of the deterministic time hierarchy theorem (their
Theorem~12.9); that proof of HU can be found as pages
297--298 of that book, and that proof itself
draws also on the framework of that book's earlier
proof of the deterministic space hierarchy
theorem.  We will assume that the reader is familiar with those
deep, classic proof presentations; anyone not expert in complexity 
will want to either first master that proof framework or just skip
the present proof.

So, suppose we are given $T_1$ and $T_2$ satisfying the 
conditions of Theorem~X\@.  We'll later specify $k$, but $k$ 
must not and will not depend on $T_1$ or $T_2$.

Our goal is to show that there is an
$\fdtime[T_2(m)]$-uniform-$\{0,1\}$-$\psr$ generator, call it $f$, whose 
rule's winner set is not obtained by \emph{any}
$\fdtime[T_1(m)]$-uniform-$Z$-$\gsr$.

Generally, we just follow the overall architecture of the HU proof, but
the differences are as follows.  First, we will be computing and
diagonalizing against functions.  So our enumeration of machines will
be an enumeration of all function-computing Turing machines (i.e.,
machines with an output tape such that when they halt, whatever is 
on the output tape is viewed as being the output).  Like HU,
we will assume that our enumeration of such machines has the property
that if $w$ encodes a machine, then $1^k w$ encodes exactly the same
machine; ``padding'' 
by adding $1^*$ as a prefix to a machine's coding does not 
change the machine encoded.  (Unlike HU, who truly need this 
due to their making a liminf claim, we merely need each machine to 
appear infinitely often in the enumeration, since for other reasons 
we are simply making a limsup claim.  However, this padding property
certainly is a fine way of achieving 
what we 
need.)

Another difference is that we need to be a pure scoring rule.  So each
length rule must appropriately link to and extend the rule from the
previous length.  As mentioned in the text immediately after
Theorem~\ref{t:hierarchy}, doing so in the obvious fashion would seem
to add a multiplicative factor of $m$, but instead we use the trick
mentioned there to avoid this.  That is, for each odd natural number
$m$, $f(0^m)$ will be $1^{\lfloor m/2\rfloor}0^{\lfloor m/2\rfloor +
  1}$ (e.g., $(0)$, $(1,0,0)$, $(1,1,0,0,0)$, etc.).  This is purely
mechanical, and does not involve any diagonalizing.  (One might 
worry that for some very small values of $m$ we might not even have time 
to realize that $m$ was odd and write the appropriate output;
that worry is because one might worry that for some small $m$ 
we may not have time greater than $m+1$ and that is not enough time to 
both see what $m$ is, and that it is odd, and to already just 
before halfway through
it have known that we are just before halfway through it and to have 
switched what we are outputting from 1s to 0s.  However, the 
$(\forall m)[T_2(m) \geq k(m+1)]$ assumption of Theorem~X will
ensures us that we do have that time.)

Now, at each even length $m$, we'll try to do a diagonalization, if we 
have time.  Our framework is that we will always have as our vector 
at this length either 
$1^{\lfloor m/2\rfloor+1}0^{\lfloor m/2\rfloor + 1}$ or 
$1^{\lfloor m/2\rfloor}0^{\lfloor m/2\rfloor + 2}$
(e.g., at length 4, either 1100 or 1000; note that 
in this proof we will quietly go back and forth notationally
between, for example,
$(1,1,0,0)$ and $1100$ and $1^20^2$).
Basically, we want to ensure that if whatever 
$\fdtime[T_1(m)]$-uniform-$Z$-$\gsr$ generator we (``we'' are the 
function $f$) at stage $m$ are 
trying to diagonalize against outputs
a vector whose winner set is the same as that given by 
$1^{\lfloor m/2\rfloor+1}0^{\lfloor m/2\rfloor + 1}$
then we will output the vector 
$1^{\lfloor m/2\rfloor}0^{\lfloor m/2\rfloor + 2}$
and otherwise we will output the vector 
$1^{\lfloor m/2\rfloor+1}0^{\lfloor m/2\rfloor + 1}$.
To do this, we immediately write onto our output tape 
the vector
$1^{\lfloor m/2\rfloor+1}0^{\lfloor m/2\rfloor + 1}$.
We do so so that if we run out of time in our 
diagonalization on this input, we at least have 
one of the two legal vectors that keep our rule pure within 
our framework.  And we do have time to write this vector,
thanks to the 
$(\forall m)[T_2(m) \geq k(m+1)]$ assumption of Theorem~X.

Now, as to the diagonalization, it goes as follows.  If $m \geq 1$ is even,
then strip away the leading 1 of $m$ in binary and the trailing
0 of $m$ in binary, and call the string
that remains $w$.  (We strip the trailing 0 because we diagonalize only
at even $m$, but we want the entire set of such strings created to 
exactly equal $\{0,1\}^*$.)
We will view $w$ as the encoding of a Turing
machine from our enumeration of function-computing TMs.  We'll run
it (unless we run out of time: as per the entire HU framework,
we'll be using a separate tape and the fully time-constructible nature of 
$T_2$ to enforce a time cutoff; actually, we'll use two 
separate tapes, since our time cutoff is $\max(T_2(m),k(m+1))$) 
on input $0^m$
to get the length-$m$ vector it outputs,
call it $v$ (if it halts and has the wrong length output, then 
it clearly isn't 
even a valid GSR\@).
We must then efficiently evaluate whether that vector has the 
same winner set as would 
$1^{\lfloor m/2\rfloor+1}0^{\lfloor m/2\rfloor + 1}$.
One might think to do so we would have to normalize $v$, which 
includes gcd's and other time-eating computations, but that is not so.  
To tell if $v$ has the same winner set as 
$1^{\lfloor m/2\rfloor+1}0^{\lfloor m/2\rfloor + 1}$, we need only
test whether $v$ is of the form 
$a^{\lfloor m/2\rfloor+1}b^{\lfloor m/2\rfloor + 1}$, 
where $a$ and $b$ are members of $Z$ and $a>b$.  If so they have the 
same winner set, and if not they don't.  Ignoring for the 
moment the cost of getting $v$, this test is 
an easy linear-time 
test on a multihead, multitape TM\@.
If we find that $v$ does have the same winner set as 
$1^{\lfloor m/2\rfloor+1}0^{\lfloor m/2\rfloor + 1}$, 
then we on our output tape just overwrite the 
$\lfloor m/2\rfloor$th character, changing it from a 1 to a 0
(for example, if $m =4$, we'd overwrite 
the second bit to change our placeholder 1100 into 
1000).  So, if we have time for our 
simulation of $w$ to complete and to evaluate whether the $v$ 
has the same winner set as 
$1^{\lfloor m/2\rfloor+1}0^{\lfloor m/2\rfloor + 1}$, and if 
it is, have time to overwrite the one bit, then we have successfully
diagonalized against the machine $w$.  

On the other hand, what if we run out of time?  No problem.  If the
machine associated with $w$ happens to 
run in time $T_1(m)$ for all $m$
(and
we're doing all $w$, so some will have that time behavior and some
won't), then---since $w$ will appear infinitely often again in our
construction, as it will also appear as $1w$, as $11w$, and so 
on---%
due to the limsup assumption we eventually
\emph{will} have
the time to fully run the diagonalization (the simulation and our
on-the-cheap comparison of winner sets and our bit-fixing).  The reason
is that the log multiplicative overhead in the limsup is (see HU)
enough to do the machine simulation, and the limsup ensures that for
any constant $c$ our $T_2$ for every $m$ beyond some point will
satisfy $T_2(m) > c T_1(m)\log T_1(m)$.\footnote{Readers familiar with
  the time hierarchy theorem may wonder why we here use a limsup but
  the HU theorem uses a liminf.  The reason is subtle.  Briefly put,
  in HU, once a machine $w$ is being diagonalized against at some
  length, it will be (in its sibling forms $1w$, $11w$, etc.)\
  diagonalized against at \emph{every} greater length.  So a liminf
  suffices, as liminf ensures that at \emph{some} (indeed, 
  infinitely many, though they don't really need 
  that) 
  greater lengths we get
  enough time overhead to diagonalize, and that is all HU needs.  
  (Warning: The HU proof itself
  sets up the right machinery and construction for liminf to suffice,
  but then inside its proofs, seems to forget this and by such
  phrases as ``infinitely often'' on its p.~298 line~8 and ``has
  arbitrarily long encodings'' on page 298 and 299, say things that
  would require limsup to support.  However, the HU construction
  supports liminf, and one can change the two mentioned phrases to
  ``almost everywhere'' and ``at every length onward once it first
  occurs'' and the thus adjusted proof becomes correct.)  In contrast,
  PSRs take as input $0^m$, and so at each length are diagonalizing
  against just \emph{one} Turing machine.  And so we do know that once
  we face off against $w$ (say, on input $0^{(1w)_{binary}}$), we'll
  infinitely often see its identical siblings at longer lengths (such
  as on input $0^{(11w)_{binary}}$, but note that that is about twice
  as long, not one longer), but we don't know we'll have them at every
  length.  Basically, since our input is $0^m$, we get to attempt just 
  one diagonalization at each length.
  Happily, under a limsup assumption, we're safe and fine, since 
  at every sufficiently long length, we have the headroom to 
  diagonalize.}

Let us discuss the constant $k$.  $k$ is simply there because for
small values of $m$, we may not have enough time to see the input, and
write the appropriate starting string to our output tape.  (All we
know is we have at least $m+1$ steps.  That isn't enough, as out input
is of the form $0^m$ so until we hit the right-end marker, all we see
is $0$s, but the starting string we need to put on the output tape as
a placeholder is roughly half 1's and half 0's, so if we just get
$m+1$ steps, we don't know roughly halfway through the 0's that we are
roughly halfway through and that we have to switch from outputting 1's
to outputting 0's.  The fact that asymptotically we have at least
$T_1(m)\log T_1(m)$ time, and thus at least $m\log m$ time, doesn't
help us for small values of $m$, since if we run out of placeholder
time at even one value of $m$ we've violated purity of our alleged PSR\@.)
However, we can set a value of $k$ that allows us to do the placeholder
writing (and on any natural TM framework, it will be a very small constant
$k$; go to the end of the string $0^m$, see if it is even, 
and then output the right placeholder string, for example by having 
on an extra dummy tape basically build
a string of just under $m/2$ 1's that we 
can use to trigger our switch-over point as we write the placeholder
output), for even small values.  Note that $k$ depends just on the 
TM framework, not on $T_1$ or $T_2$.  

All that remains is how to move from our proof of Theorem~X to a proof
of Theorem~\ref{t:hierarchy}.  The previous paragraph's contortions
will actually already hint to the reader the path to doing this.
Namely, suppose we now must operate without the crutch of having
$(\forall m)[T_2(m) \geq k(m+1)]$.  Our salvation is the limsup claim,
which ensures that eventually, $T_2$ becomes very big.  So all we need
to do is find a way to mark time until that kicks in, and by marking
time, we have to make do with $m+1$ steps as that is all we know we
always have.  We can't such marking time within the framework of the
particular PSR we've been using above, because as noted above, one
would seem to have to be able to know the midpoint of $0^m$ in real
time as one was passing it, and that is impossible.  The workaround 
is as follows.  We change the PSR that we'll use.  In particular, for
each $T_1$ and $T_2$, there exists some odd value $m'$ beyond which the
limsup ensures that we do have easily enough time headroom.  So what 
we'll do is our PSR will simply be $0^m$ at each length $m \leq m'$ 
(note that we can do this in time $m+1$, at length $m$), and 
{}from then on, it will mimic our trick except with those 
extra 0's, namely, we will fixing 
our vector at each odd length, $m > m'$, to be
$1^{\lfloor (m-m')/2\rfloor}0^{\lfloor (m-m')/2\rfloor + 1 + m'}$,
and at even lengths $m > m'$, our vector will be either 
$1^{\lfloor (m-m')/2\rfloor + 1 }0^{\lfloor (m-m')/2\rfloor + 1 + m'}$
or 
$1^{\lfloor (m-m')/2\rfloor}0^{\lfloor (m-m')/2\rfloor + 2 + m'}$.
This allows us to remove the requirement that
$(\forall m)[T_2(m) \geq k(m+1)]$, and so we have established 
Theorem~\ref{t:hierarchy}.
Note that $m'$ will depend on $T_1$ and $T_2$, but that is 
perfectly legal; it does not undermine the proof.%
\end{proof}

We finish with some brief, technical comments about possible
extensions and alternate proofs.  First, our diagonalization is carried
out the world that it is in, namely, the world of functions.  We
mention in passing that since the PSR we build is variable only as
to which of two possible vectors it has at each even length, the
information behind its choice is in effect a set---indeed, a set of the
form $A \subseteq (11)^*$, and so a particular type of tally set.  One
could use this to make the proof more set-focused. But one would still
have to account for overhead time and so on, and we feel it is more
natural to simply have the proof be in the world the items in question
are inhabiting.  Second, we mention that it is possible that one could
turn Theorem~\ref{t:hierarchy}'s limsup into a liminf, thus making the
theorem 
stronger.  However, for most natural time functions $T_1$ and
$T_2$, the liminf and the limsup are equal, and so we do not
consider this an important direction, especially as it would likely make the
proof far more complex if attempted.  We do mention that the natural path
to use to attempt such an improvement would be, once one first was
trying to diagonalize against a given machine, to 
attempt diagonalizing against that machine
at every
length from then on until one had successfully diagonalized
against it.  But we
warn that to even know what one was currently trying to diagonalize
against would require recomputing one's history, which itself takes
time and can interfere with the proof.  
Worse, our current setup
doesn't do any diagonalizations at even lengths, but the 
liminf could go to zero just due to even lengths, and
that would not at all help us \emph{ever} have time 
to diagonalize; so our entire
zigzag framework becomes poisonous, yet abandoning it loses its
advantage (already lost anyway in the type of new construction we are
speculating about) of avoiding history rebuilding.

\subsection{Proof of Theorem~\ref{t:all-differ}}

In the following theorem, nonuniform simply means that the corresponding generators
are not required to be computable in any specific complexity class. In fact, 
the generators may be uncomputable.

\talldiffer*

\begin{proof}
Let us first show that 
\textbf{\boldmath $\nonnegrationals$-$\psr \not\subseteq$ nonuniform-$\integers$-\psr}.
Consider the following $\fp$-uniform $\nonnegrationals$-\psr\ generator $f$.

\noindent
$f(0^1) = (1)$, \\
$f(0^2) = ({3 \over 2}, 1)$, \\
$f(0^3) = ({3 \over 2}, 1, {1 \over 2})$, \\
$f(0^4) = ({7 \over 4}, {3 \over 2}, 1, {1 \over 2})$, \\
$f(0^5) = ({7 \over 4}, {3 \over 2}, 1, {1 \over 2}, {1 \over 4})$, \\
$f(0^6) = ({15 \over 8}, {7 \over 4}, {3 \over 2}, 1, {1 \over 2}, {1 \over 4})$, \\
$f(0^7) = ({15 \over 8}, {7 \over 4}, {3 \over 2}, 1, {1 \over 2}, {1 \over 4}, {1 \over 8})$, \\
etc.

Consider a nonuniform $\integers$-$\psr$ type generator $g$ that claims to
have the same winner set as this on all instances.  By the (clear)
extension of Proposition~\ref{p:differ} mentioned at the end of the
paragraph that follows that result, 
the 5-candidate vector of $g$
must normalize to
$(6,5,3,1,0)$.  The ordered gap pattern in this is: 1,2,2,1.
So the actual 5-candidate vector (over $Z$) of $g$ satisfies
$2(\alpha^5_1 - \alpha^5_2) = 
(\alpha^5_2 - \alpha^5_3) = 
(\alpha^5_3 - \alpha^5_4) = 
2(\alpha^5_4 - \alpha^5_5)$.
Similarly, the 7-candidate vector of $g$ must normalize to a 
vector having gap pattern: 1,2,4,4,2,1.  Since that vector 
retains all 5 coefficients from the 5-candidate vector, with 
two additional coefficients appropriately added, by inspection of 
the possibilities, it is clear that the only possible extension
that achieves this is one that adds exactly the coefficients
$\alpha^5_1 + 
{1 \over 2}(\alpha^5_1 - \alpha^5_2)$
and 
$\alpha^5_5 - 
{1 \over 2}(\alpha^5_1 - \alpha^5_2)$ (and due to the 6-candidate
vector, they must be added in the order just stated).  That is, 
the vector extends to the right and left, but an amount half the 
``base'' gap.  However, note that if 
$\alpha^5_1 - \alpha^5_2$ is not a multiple of 2, 
this extension is already not legal over $\integers$, as it would have 
coefficients not in $\integers$.  If 
$\alpha^5_1 - \alpha^5_2$ is a multiple of $2$, then 
we indeed do have this valid extension of the 5-candidate vector.
But by the same argument, the only possible 9-candidate vector 
again will have to be formed from the 7-candidate vector by 
extending to the right and left ends by half the then-current ``base''
gap.  Note that that will not be possible unless 
$\alpha^5_1 - \alpha^5_2$ is a multiple of $2^2 = 4$.
And so on.  Since there is some power of 2 that is the maximum 
power of two that divides 
$\alpha^5_1 - \alpha^5_2$, this process will eventual end 
in failure, i.e., there will be no valid vector over $Z$ that 
properly extends the pattern $g$ has trapped itself into.  Thus,
contrary to the claim made for $g$, $g$ in fact cannot match 
the winner sets of $f$.

Let us now show that 
\textbf{\boldmath $\integers$-$\psr \not\subseteq$ nonuniform-$Q'$-\psr}.
We will give the argument for this case more compactly, since 
the reader by now will be familiar with the flavor of arguments 
of this sort, from the above case.
Consider the following $\fp$-uniform $\integers$-\psr\ generator $f$.

\noindent
$f(0^1) = (0)$, \\
$f(0^2) = (1,0)$, \\
$f(0^3) = (1,0,-1)$, \\
$f(0^4) = (2,1,0,-1)$, \\
$f(0^5) = (2,1,0,-1,-2)$, \\
$f(0^6) = (4,2,1,0,-1,-2)$, \\
$f(0^7) = (4,2,1,0,-1,-2,-4)$, \\
etc.

Consider a nonuniform $\nonnegrationals$-$\psr$ type generator $g$ that claims to
have the same winner set as this on all instances.  Note that the gap
pattern, at each length starting at 5, is that we have four identical,
adjacent smallest-size (among the gaps that exist) gaps in the center,
and then surrounding that the gaps grow by repeatedly doubling (with
the extra one being on the ``upper'' side on even-length cases).  Note
that no single internal insertion anywhere maintains a four, identical,
adjacent, smallest-size (among the gaps that exist) gap pattern.  So we 
can only expand by going doubly-far to the right and left in each 
pair of extensions.  But that means the values $\alpha^{2m+1}_{2m+1}$ will
be decreasing at an exponentially increasing rate.  And so since 
$\alpha^1_1$ was finite, at some $m$
the value 
$\alpha^{2m+1}_{2m+1}$ will necessarily have to be less than zero, and 
so will not be an element of $\nonnegrationals$.  

We'll sketch 
\textbf{\boldmath $\rationals$-$\psr \not\subseteq$ nonuniform-$\nonnegrationals$-$\psr \cup {}$ nonuniform-$\integers$-$\psr$} even more briefly.  
It does not automatically follow from the earlier parts; claiming 
that would be like claiming that since not all integers are even 
and not all integers are odd, it follows that not all integers are 
odd integers or even integers.  Nonetheless, by building a 
\emph{construction} that uses both of the weaknesses exploited by 
the constructions of the preceding two parts, we can easily establish
that 
\textbf{\boldmath $\rationals$-$\psr \not\subseteq$ nonuniform-$\nonnegrationals$-$\psr \cup {}$ nonuniform-$\integers$-$\psr$}.
In particular, it is easy to prove that the following $\fp$-uniform
$\rationals$-\psr\ has the desired property.  After setting down 4 equally spaced 
gaps, it in term makes re-doublingly large gaps going in the negative 
direction 
and re-halvingly large gaps going in the positive direction.

\noindent
$f(0^1) = (2)$, \\
$f(0^2) = (2,1)$, \\
$f(0^3) = (2,1,0)$, \\
$f(0^4) = (2,1,0,-1)$, \\
$f(0^5) = (2,1,0,-1,-2)$, \\
$f(0^6) = (2,1,0,-1,-2,-4)$, \\
$f(0^7) = (2{1 \over 2},2,1,0,-1,-2,-4)$, \\
$f(0^8) = (2{1\over 2},2,1,0,-1,-2,-4,-8)$, \\
$f(0^9) = (2{3 \over 4}, 2{1\over 2},2,1,0,-1,-2,-4,-8)$, \\
$f(0^{10}) = (2{3 \over 4}, 2{1\over 2},2,1,0,-1,-2,-4,-8,-16)$, \\
$f(0^{11}) = (2{7\over 8},2{3 \over 4}, 2{1\over 2},2,1,0,-1,-2,-4,-8)$, \\
etc.

By using the natural 
variants for this of 
both types of arguments used in the earlier parts 
of this proof, we can argue that this function 
can be accepted by neither a $\nonnegrationals$-type generator nor a $\integers$-type generator.%
\end{proof}

\subsection{Proof of Theorem~\ref{t:fpsr}}

\tfpsr* 

\begin{proof}
Theorem~\ref{t:fpsr} is immediately clear, 
in light of Proposition~\ref{p:differ} and the end of 
the paragraph following it, 
since our normalizations 
shifted each of these types to equivalent vectors over $\naturals$, which is 
the most restrictive of all these types.%

The equality with $\rationals$-$\psr$ can easily be seen by inductively constructing, for every flexible generator $f$, an equivalent pure generator over the rationals.
\end{proof}

\subsection{Proof of Theorem~\ref{t:norm-zero}}

\tnormzero* 

\begin{proof}
Let us prove the first part of the theorem.  Consider this 
$\fp$-uniform $\naturals$-\psr, which we mention in passing is even 
an 
$\fp$-uniform $\naturals$-$\psr_{\normgcd}$.
$f(0^1) = (5)$, 
$f(0^2) = (6,5)$, 
$f(0^3) = (7,6,5)$, 
$f(0^4) = (8,7,6,5)$, and
$f(0^5) = (8,7,6,5,0)$.
We won't use longer lengths here, but to make clear that 
this is a pure scoring rule, let us say that the remaining
lengths are $f(0^6) = (8,7,6,5,0,0)$, $f(0^7) = (8,7,6,5,0,0,0)$, 
and so on.

Consider an $\naturals$-$\psr_{\normzero}$ type generator $g'$ that claims 
to have the same winner set as this on all instances.  So by 
Proposition~\ref{p:differ}, it is clear that the 4-candidate vector 
of $g'$ must be of the form $(3k,2k,k,0)$, for some $k \in \naturals - \set{0}$.

But clearly addition of one coefficient to this vector cannot yield 
anything that normalizes to $(8,7,6,5,0)$; one can see this easily by
examining each of the legal insertion places and seeing that it cannot 
suffice.  So the rule of $g'$ differs from the original rule on some 
5-candidate examples, by Proposition~\ref{p:differ}.

(We mention in passing that the exact same $f$ we just used can 
also be easily argued 
to be accepted by no $\psr$ generator that has 
the property that at each length, the gcd of its nonzero coefficients
(if any) is one.)

Let us turn to the second claim of the theorem, namely, that 
every $\naturals$-$\psr$ is generated by an $\naturals$-$\psr$ 
generator that on all but at most 
a finite number of lengths $m$ has the property that the last coefficient
is zero and the gcd of the nonzero coefficients in the length-$m$ vector 
is one.

Given $\cale_1$, an $\fp$-uniform $N$-\psr, fix an $\fp$ generator 
for that rule.  Let $\pool = \{ \alpha^m_j \condition  m \geq 1 \land
1 \leq j \leq m\}$, i.e., it is the set of all coefficients.

If $\|\pool\| = 1$ we are done easily using the scoring vector 
sequence $(0)$, $(0,0)$, $(0,0,0)$, \ldots.  

So henceforth we assume that 
$\|\pool\| \geq 2$.  Let $\min(\pool) = k$.  Now, alter the generator
from now on (and thus each $\alpha^i_j$, and so indirectly
we also have altered $\pool$)
to subtract $k$ from each original $\alpha^i_j$.  
So it is clear that, after the alteration, at all but a finite number 
of values of $m$ we have $\alpha^m_m = 0$.

Let  $j = \gcd(\pool)$, for the (altered) $\pool$.

Let us keep in mind that with PSRs, coefficients that occur at a
length must also occur at each greater length.  If $j=1$, then there
is some finite length at which the set of coefficients at that length
has a gcd of one, and so at all greater lengths the coefficient set
also has a gcd of one.  If $j \geq 2$, then further re-alter our
generator to, after subtracting $k$, divide by $j$.  Applying the
$j=1$ argument to this, we again have that the gcd will be one at all
but a finite number of lengths, and as dividing by $j$ does not change
the set of lengths at which $\alpha^m_m=0$, we also have that for
our re-altered generator $\alpha^m_m$ is zero for all but a finite
number of $m$'s.  The re-altered generator is clearly polynomial-time 
computable if the original one is, so we have completed the proof of
the theorem's second part for the $\fp$-uniform case, and the result 
clearly also holds for the nonuniform case, by the same argument with 
the discussion of runtimes removed.%
\end{proof}

\fi

\ifshowcomplexityappendix
\section{Omitted Proofs and Discussion from Section~\ref{sect:dichotomy}}\label{appendix:dichotomy}

\subsection{Missing Pieces of the Proof of Theorem~\ref{theorem:alpha,beta,zero in ptime}}

Here we present the proofs of the two facts used in the proof of Theorem~\ref{theorem:alpha,beta,zero in ptime}. We start with Fact~\ref{fact:extending vi voters in ptime}:

\extendingvivotersinptime*

\begin{proof}
  For each candidate $c\in C$, let $s_c$ be the score of that candidate after the voters in $V$ and the voters in $R$ have been counted. For the remainder of the proof of Fact~\ref{fact:extending vi voters in ptime}, we assume that $i=1$, i.e., we are adding voters with $p$ in the first position. Since the proof of Fact~\ref{fact:extending vi voters in ptime} does not use the fact that $\alpha\ge\beta$, the proof works in the same way for $i=2$, i.e., adding voters voting $p$ in the second position.
  
  We can add at most $r=k-\card{S}$ voters from $V_1$. Clearly, $k$ is bound by the size of the instance. Thus we can brute-force over every $j\in\set{0,\dots,r}$ and check whether $S$ can be extended with exactly $j$ voters from $V_1$. To do this for a given $j$, we proceed as follows:
  
  \begin{itemize}
   \item First, we compute the number of points that $p$ will have after adding $j$ voters from $V_1$, this is $s_p+j
\alpha$. We denote this number with $s^{final}_p$.
   \item For each candidate $c\in C$, we compute the maximal number $m_c$ of votes giving points to $c$ that can be added. This is the maximal number $m_c$ with $s_c+m_c
\beta\leq s^{final}_p$. If $m_c$ is negative for some $c$, then $S$ cannot be extended with exactly $j$ voters from $V_1$.
   \item For each candidate $c\in C$, we add voters voting $c$ in the second place while still possible, i.e., while all of the following hold:
   \begin{itemize}
     \item the number of added voters up to now does not exceed $j$,
     \item the number of added voters voting $c$ second does not exceed $m_c$,
     \item there are still voters available voting $c$ second.
   \end{itemize}
   \item If the resulting set is a solution to the control problem, accept. Otherwise, $S$ cannot be extended with exactly $j$ voters from $V_1$.~\qedhere
  \end{itemize}
 \end{proof}
 
The proof for Fact~\ref{fact:we mostly add V1 voters} follows:

\wemostlyaddVOnevoters*

 \begin{proof}
  Let $S_1$ ($S_2$) be the corresponding subset of $V_1$ ($V_2$). Let $c_1,\dots,c_\ell$ be the candidates voted in the second place by the voters in $S_1$. We compare the effect that adding $S_1$ or $S_2$ has on the relationship between $p$ and the $c_i$ (clearly, for the relationship between $p$ and candidates $c\notin\set{c_1,\dots,c_\ell}$, adding $S_1$ is always better).
 
 \begin{itemize}
  \item When adding $S_1$, $p$ gains $(\ell-1)
\alpha + (\alpha-\beta)=\ell\alpha-\beta$ points against each $c_i$.
  \item When adding $S_2$, $p$ gains at most $\ell
\beta$ points against each of the $c_i$.
 \end{itemize}

 Since $\ell(\alpha-\beta)\ge\beta$, we have that $\ell\alpha-\beta\ge \ell\beta$, and so adding $S_1$ is at least as good as adding $S_2$.
 \end{proof}

\subsection{Proof of Lemma~\ref{lemma:coefficients realization}}

The proof Lemma~\ref{lemma:coefficients realization} makes extensive use of the following construction:

\begin{lemma}\label{lemma:transfer difference alpha points}
 Given a scoring vector $(\alpha_1,\dots,\alpha_m)$ , and numbers $i\neq j$, $k\neq l\in\set{1,\dots,m}$, there is a set $V$ of $m$ votes such that when counting the votes in $V$, the scores of all candidates is $A:=\sum_{m=1}^n\alpha_m$, except:
 \begin{itemize}
  \item candidate $c_i$ has $A+\alpha_k-\alpha_l$ points,
  \item candidate $c_j$ has $A+\alpha_l-\alpha_k$ points.
 \end{itemize}
\end{lemma}

\begin{proof}
 Without loss of generality, assume $i\leq j$ and $k\leq l$, otherwise we swap. 
 
 Let $v_1$ be a vote having $c_i$ in position $l$ and $c_j$ in position $k$, the remaining candidates are ordered arbitrarily. For any $\ell\ge1$, let $v_{\ell+1}$ be the vote obtained from $v_\ell$ by moving each candidate one position to the right, i.e., the candidate at position $t$ in the vote $v_\ell$ ends at position $t+1$ in the vote $v_{\ell+1}$ if $t+1\leq m$, and at position $1$ if $t+1=m$. Now let $V_0$ contain the votes $v_1,\dots,v_m$. Clearly, each candidate gains exactly $A$ points in the votes $V_0$. Now, let $V$ be obtained from $V_0$ by exchanging, in the vote $v_0$, the positions of candidate $c_i$ and $c_j$. Then, relatively to the votes $V_0$, the points of $c_i$ and $c_j$ change as follows:
 
 \begin{itemize}
  \item $c_i$ loses $\alpha_l$ points and gains $\alpha_k$ points,
  \item $c_j$ loses $\alpha_k$ points and gains $\alpha_l$ points. 
 \end{itemize}
 
 So the relative score is as required.
\end{proof}

The actual proof of Lemma~\ref{lemma:coefficients realization} follows:

\coefficientsrealization*

\begin{proof}
 The algorithm produces the set $V$ by consecutive applications of Lemma~\ref{lemma:transfer difference alpha points}. We say that an application of the Lemma \emph{transfers} $(\alpha_k-\alpha_l)$ points from $c_j$ to $c_i$, since the relative score (in relation to all candidates except $c_i$ and $c_j$) of $c_i$ increases by $\alpha_k-\alpha_l$, while the score of $c_j$ decreases by the same amount. Note that in order for such a ``transfer'' to be possible, neither candidate is required to actually have a nonzero score. We start with the empty set $V$, and then, for each $a^c_i\neq 0$ for some $c\in\set{1,\dots,m-1}$, we proceed as follows:
 \begin{itemize}
  \item if $a^c_i>0$, we apply Lemma~\ref{lemma:transfer difference alpha points} $a^c_i$ times to transfer $\alpha_i$ points from the candidate $m$ to the candidate $c$,
  \item if $a^c_i<0$, we apply Lemma~\ref{lemma:transfer difference alpha points} $(m-2)
 a^c_i$ times to transfer $\alpha_i$ points from the candidate $m$ to each candidate $c'\neq\set{c,m}$.
 \end{itemize}
 In the above steps, clearly the score of $m$ is the lowest among all candidates. So to ensure that $\score{i}>\score{m}+k
\alpha_1$, it suffices to transfer $(k+1)
\alpha_1$ points from the candidate $m$ to each candidate $c\in\set{1,\dots,m-1}$. Clearly, the resulting scores are as required.
\end{proof}

\subsection{Proof of \ifshortenedmaintext the ``Dual'' of Theorem~\ref{theorem:3-veto generalization} \else Theorem~\ref{theorem:4-approval generalization informal} \fi and Discussion of ``Partial Reductions''}\label{subsect:4 approval generalization and partial reduction discussion}

In
\lhnote{LATER EDIT: I have reedited this paragraph somewhat.}%
 the discussion surrounding Theorems~\ref{theorem:3-veto generalization}
and~\ref{theorem:alpha4 > alpha(n-2) np completeness}, we described 
Theorem~\ref{theorem:3-veto generalization} and its ``dual''---which 
we now state \ifshortenedmaintext\else formally\fi as
Theorem~\ref{theorem:4-approval generalization}---as ``partial reductions.''
To clarify this unusual terminology, we now explain why we use this term.
For a generator $f$ that is covered in Theorem~\ref{theorem:alpha4 > alpha(n-2) np completeness},
it is not necessarily the case that hardness follows from Theorem~\ref{theorem:3-veto generalization}
alone (or from Theorem~\ref{theorem:4-approval generalization} alone). 
Rather, for each value of $m$ resulting from a 3DM-instance, the reduction establishing NP-hardness
(which is given in the proof of Theorem~\ref{theorem:alpha4 > alpha(n-2) np completeness})
needs to check which of the two cases 
(as to meeting the differing ``$\alpha_{\dots} < \alpha_{\dots}$'' 
conditions listed in Theorem~\ref{theorem:3-veto generalization}
and Theorem~\ref{theorem:4-approval generalization}) is satisfied.
\lhnote{I asked Henning:
\begin{verbatim}
can you please suggest an alternate phrase to

(meeting the prerequisites of Theorem~\re\f{theorem:3-veto generalization}
or Theorem~\ref{theorem:4-approval generalization}) is met

for the discussion at the start of B.3?  "prerequisites" doesn't really work and the reader will have trouble knowing precisely what it refers to.  but i'm not sure what alternate phrase would correctly capture the semantics.    is "prereqs" referring precisely and only to the ">" of the second bullet of 3.8 and the ">" of the 3rd bullet of B.2?  if so, can we safely replace the above phrase with this

(as to meeting the differing ``$\alpha_{\dots} < $\alpha_{\dots}$'' conditions listed in Theorem~\re\f{theorem:3-veto generalization}
and Theorem~\ref{theorem:4-approval generalization}) is satisfied.
\end{verbatim}

And he replied

\begin{verbatim}
I also thought about this and kept it this way because at that point the
reader probably cares more about which reductions (theorems) can be
applied than which of the ">" conditions are met (and yes, this is only
the > condition). If you think just giving the > condition explicitly
then that's fine for me as well. On further reading the paragraph that
probably is the best option. (You may want to mention that the $m$
referred to here is simply $6k$, where $k$ is the size of the 3DM instance.)

Another possibility would be to rename the reduction functions, call one
of them g_4APP and the other g_3VETO and then talk about which of the
reductions is applicable (but that's probably too much change for now).
\end{verbatim}

and I wrote back, and did, this:

\begin{verbatim}
thanks.  so for now i'll just change it to this:

(as to meeting the differing ``$\alpha_{\dots} < $\alpha_{\dots}$'' conditions listed in Theorem~\re\f{theorem:3-veto generalization}
and Theorem~\ref{theorem:4-approval generalization}) is satisfied.


(note: I'll really just include \dots in the source probably---i think that is enough to tell them exactly what part of 3.8 and B.2 we are speaking of.)

i'm done as soon as i edit that in, and so will release the lock without 5-15 minutes, along with a complete diff.
\end{verbatim}
}%
  Depending on which condition applies, either the
reduction $g$ from Theorem~\ref{theorem:3-veto generalization} or the reduction $g$ from
Theorem~\ref{theorem:4-approval generalization} is used for this
particular 3DM-instance.
In this sense, Theorem~\ref{theorem:3-veto generalization} (and also its ``dual,'' Theorem~\ref{theorem:4-approval generalization})
are ``partial'' reductions: For a given $f$, each of 
these theorems may only be ``half''
\lhnote{LATER EDIT: 
quotation marks added as 
protection against mean referees---because 
we are NOT making a quantitative claim about 
one many $m$ one or the other holds for asymptotically---half means 
one of two reductions... it is not quantitative over the m's.}%
of the reduction needed to prove the NP-hardness of $f$-CCAV. 
Theorems~\ref{theorem:3-veto generalization} and~\ref{theorem:4-approval generalization}
together, with a simple additional observation, then give
Theorem~\ref{theorem:alpha4 > alpha(n-2) np completeness}, which is the
``complete'' reduction to obtain NP-completeness of $f$-CCAV.

Let us turn from the above discussion to the formal statement of
\ifshortenedmaintext
the ``dual'' of Theorem~\ref{theorem:3-veto generalization}
\else
Theorem~\ref{theorem:4-approval generalization informal}
\fi
and its proof:
\ifshortenedmaintext
\else
We recall the statement of the Theorem:

\fourapprovalgeneralizationinformal*

\fi
The complete statement of \ifshortenedmaintext the analogue of Theorem~\ref{theorem:3-veto generalization} \else Theorem~\ref{theorem:4-approval generalization informal} \fi is as follows. Note that again, we do not require the generator $f$ to respect any purity conditions.

\begin{theorem}\label{theorem:4-approval generalization}
 Let $f$ be an \fp-uniform $\rationals$-generator. Then there exists an \fp-computable function $g$ such that 
  \begin{itemize}
  \item $g$ takes as input an instance $I_{\mathtext{3DM}}$ of 3DM where $k=\card X=\card Y=\card Z$,
  \item $g$ produces an instance $I_{\mathtext{CCAV}}$ of $f$-CCAV,
  \item If for $m=6k$ we have that $\alpha^{f,m}_4>\alpha^{f,m}_{m-3k+1}$, then: $I_{\mathtext{3DM}}$ is a positive instance of 3DM if and only if $I_{\mathtext{CCAV}}$ is a positive instance of $f$-CCAV.
 \end{itemize}
\end{theorem}

\begin{proof}
The proof is very similar to the proof of Theorem~\ref{theorem:3-veto generalization} above, in particular we use the same set of candidates. Again, let $M\subseteq X\times Y\times Z$ be the set from $I_{\mathtext{3DM}}$, and let $X=\set{s_1,\dots,s_k}$, $Y=\set{s_{k+1},\dots,s_{2k}}$, $Z=\set{s_{2k+1},\dots,s_{3k}}$. Similarly to the proof of theorem~\ref{theorem:3-veto generalization}, we use Lemma~\ref{lemma:coefficients realization} to construct an election with relative points as follows:
 
 \begin{itemize}
  \item $\score{p}=0$,
  \item for each $i\in\set{1,\dots,3k}$, $\score{s_i}=k
\alpha_1-(k-1)
\alpha_{m-3k+i}-\alpha_{1+r(i)}$,
  \item all dummy candidates have points such that adding at most $k$ votes does not let them win the election.
 \end{itemize}

For each $(x,y,z)=(s_h,s_i,s_j)\in M$ with $h<i<j$, we add an available voter voting as follows:
 $p>x>y>z>d_4>\dots>d_{m-3k-1}>s_1>\dots>s_{h-1}>d_1>\linebreak s_{h+1}>\dots>s_{i-1}>d_2>s_{i+1}>\dots>s_{j-1}>d_3>s_{j+1}>\dots,s_{3k}$.

 We say that $x$, $y$, and $z$ are \emph{approved} in this vote. Note that a candidate $s_i$ gets $\alpha_{1+r(i)}$ (where $r(i)=1$, $2$ or $3$ depending whether $s_i\in X$, $Y$ or $Z$) points in a vote approving $s_i$, and gets $\alpha_{m-3k+i}$ in a vote not approving $s_i$. 
 
 We claim that $p$ can be made a winner by adding at most $k$ of the available votes if and only if the 3DM instance is positive.
 
 First assume that the 3DM instance is positive, and let $C\subseteq M$ be a cover with $\card{C}=k$. We add the votes corresponding to the elements of $C$, i.e., for each candidate $s_i$, we add one vote that approves $s_i$ and $(k-1)$ votes that do not approve $s_i$. The final score of the nondummy candidates is as follows:
 
 \begin{itemize}
  \item Candidate $p$ gains $\alpha_1$ points in each of the $k$ added votes, so the final score of $p$ is $k
\alpha_1$,
  \item each candidate $s_i$ gains $\alpha_{1+r(i)}+(k-1)
\alpha_{m-3k+i}$ points from the one vote approving $s_i$ and the $k-1$ votes not approving $s_i$. Thus the final score of $s_i$ is $k
\alpha_1-(k-1)
\alpha_{m-3k+i}-\alpha_{1+r(i)}+\alpha_{1+r(i)}+(k-1)
\alpha_{m-3k+i}=k
\alpha_1$.
 \end{itemize}
 
 Thus $p$ and $s_i$ tie, and so $p$ is a winner of the election.
 
 For the converse, assume that $p$ can be made a winner by adding at most $k$ voters. Let $C$ be the elements of $M$ corresponding to the added votes, then $\card{C}\leq k$. 
 
 Since $\alpha_1\ge\alpha_4>\alpha_{m-3k+1}$, we know that the score of $s_1$ before adding any votes is positive. Thus $C\neq\emptyset$, i.e., at least one vote is added. Let $z$ be a candidate from $Z$ who is approved in at least one of the added votes. We show that $\card{C}=k$. For this, indirectly assume that $\card{C}<k$. We show that $z$ strictly beats $p$ in this case: Consider the final scores of $p$ and $z$. Let $i$ be the index of $z$, i.e., the $i$ with $s_i=z$. Note that $r(i)=3$, since $s_i=z\in Z$.
 
 \begin{itemize}
  \item candidate $p$ gains $\card{C}
\alpha_1$ points,
  \item candidate $z$ gains at least $\alpha_4+(\card C-1)
\alpha_{m-3k+i}$ points (even more if $z$ is approved in more than one of the additional votes). Thus the final score of $z$ is at least
  
  $k
\alpha_1-(k-1)
\alpha_{m-3k+i}-\alpha_{4}+\alpha_4+(\card C-1)
\alpha_{m-3k+i}
   =
  k
\alpha_1+(\card C-k)
\alpha_{m-3k+i}
  $
 \end{itemize}
 
 To see that $z$ strictly beats $p$, we compute the difference between their scores, which is 
 
 $$
 \begin{array}{ll}
    & \underbrace{k
\alpha_1+(\card C-k)
\alpha_{m-3k+i}}_{\score{s_i}}-\underbrace{\card{C}
\alpha_1}_{\score{p}} \\
  = & (k-\card{C})
\alpha_1-(k-\card C)
\alpha_{m-3k+i} \\
  = & (k-\card{C})
(\alpha_1-\alpha_{m-3k+i}).
 \end{array}$$
 
 Since, by assumption, $k-\card{C}>0$, and $\alpha_1\ge\alpha_4>\alpha_{m-3k+1}\ge\alpha_{m-3k+i}$, this difference is positive and so $s_i$  strictly beats $p$ if $\card{C}<k$. Thus we know that $\card{C}=k$. 
 
 We now show that $C$ is indeed a cover. Assume that this is not the case, so, since $\card C=k$, without loss of generality, there is some $z'\in Z$ which is covered twice, i.e., which is approved in at least two of the added votes. We show that this $z'$ strictly beats $p$ in the election after the addition of votes by comparing the score of $p$ and $z'$. As above, let $i$ be the index of $z'$, i.e., chose $i$ such that $z'=s_i$. The scores are as follows (again, $r(i)=3$ since $s_i\in Z$):
 
 \begin{itemize}
  \item The candidate $p$ again gets $k
\alpha_1$ points, and so its final score is $k
\alpha_1$.
  \item The candidate $z'$ is approved in at least $2$ of the $k$ additional votes. Thus $z'$ gains at least $2
\alpha_4+(k-2)\alpha_{m-3k+i}$ points. The final score of $z'$ is thus at least
  
  $k
\alpha_1-(k-1)
\alpha_{m-3k+i}-\underbrace{\alpha_{1+r(i)}}_{=\alpha_4}
  +2\alpha_4
  +(k-2)\alpha_{m-3k+i}
  =k
\alpha_1-\alpha_{m-3k+i}+\alpha_4$.
 \end{itemize}

 Since $\alpha_4>\alpha_{m-3k+1}\ge\alpha_{m-3k+i}$, it follows that the final score of $z'$ exceeds the score of $p$, and so $z'$ indeed strictly beats $p$ as claimed if $C$ is not a cover. Thus $C$ is indeed a cover as required.
\end{proof}

\subsection{Proof of Theorem~\ref{theorem:alpha4 > alpha(n-2) np completeness}}
 
\alphafourbiggeralphanminustwonpcomplete*

\begin{proof}
 Clearly, for pure scoring rules, if the condition $\alpha^{f,m}_4>\alpha^{f,m}_{m-2}$ is true for some $m$, then it remains true for all $m'\ge m$, since it is easy to see that $\alpha^{f,m+1}_4\ge\alpha^{f,m}_4$ and $\alpha^{f,m+1}_{m+1-2}\leq\alpha^{f,m}_{m-2}$.

 We prove NP-hardness by a reduction from 3DM. So let $I_{\mathtext{3DM}}$ be a 3DM-instance, and let $k$ be the cardinality of the set $X$ in $I_{\mathtext{3DM}}$. Without loss of generality, we can assume that $k$ is large enough such that the number $m=6k$ satisfies the condition $\alpha^{f,m}_4>\alpha^{f,m}_{m-2}$. Due to the monotonicity of the coefficients and since $m-3k+1=3k+1$, we have that
 
 $$\alpha^{f,m}_4\ge\alpha^{f,m}_{m-3k+1}=\alpha^{f,m}_{3k+1}\ge\alpha^{f,m}_{m-2},$$
 
 and since $\alpha^{f,m}_4>\alpha^{f,m}_{m-2}$, we know that one of the cases $\alpha^{f,m}_{3k+1}>\alpha^{f,m}_{m-2}$, or $\alpha^{f,m}_4>\alpha^{f,m}_{m-3k+1}$ occurs. In the first case, we use the reduction from Theorem~\ref{theorem:3-veto generalization}, in the second case, the one from Theorem~\ref{theorem:4-approval generalization}.
\end{proof}

\subsection{Proof of Theorem~\ref{theorem:summary finitely many coefficients}.\ref{theorem part:alpha2>alpha4 np complete}}

\summaryfinatelymanycoefficients*

\begin{proof}
 Again, we reduce from 3DM. Let $n=\card{M}$, and let $k=\card{X}=\card{Y}=\card{Z}$. We introduce the following candidates:
 
 \begin{itemize}
  \item The preferred candidate $p$,
  \item for each $c\in X\cup Y\cup Z$ a candidate $c$,
  \item two dummy candidates $d_1$ and $d_2$.
 \end{itemize}
 
 Using Lemma~\ref{lemma:coefficients realization}, we construct the registered voters such that the relative score of the candidates before adding voters is as follows:
 
 \begin{itemize}
  \item $\score{p}=-k
\alpha_4$,
  \item $\score{d_1},\score{d_2}<-k
(\alpha_1+\alpha_4)$,
  \item $\score{x}=-\alpha_1-(k-1)
\alpha_4$ for each $x\in X$,
  \item $\score{y}=-\alpha_2-(k-1)
\alpha_4$ for each $y\in Y$,
  \item $\score{z}=-(k-1)
\alpha_4$ for each $z\in Z$.
 \end{itemize}

 For each $(x,y,z)\in M$, we add one available voter voting $$\underbrace{x}_{\alpha_1}>\underbrace{y}_{\alpha_2}>\underbrace{d_1}_{\alpha_3}>\underbrace{p>\dots}_{\alpha_4}>\underbrace{d_2}_{\alpha_5}>\underbrace{z}_0.$$
 
 We claim that $p$ can be made a winner by adding at most $k$ voters if and only if there is a cover $C\subseteq M$ with $\card{I}=k$.
 
 First, assume that there is such a cover $C$. Then, adding the voters corresponding to the elements of $C$ changes the scores as follows:
 
 \begin{itemize}
  \item $p$ gains $k
\alpha_4$ points and thus ends up with $0$ points,
  \item each $x$ ($y$) gains $\alpha_1+(k-1)\alpha_4$ points ($\alpha_2+(k-1)\alpha_4$ points), which leads to $0$ points,
  \item each $z$ gains $(k-1)
\alpha_4$ points, also finishing with $0$ points,
  \item by construction, $d_1$ and $d_2$ cannot win the election.
 \end{itemize}
 
 For the other direction, assume that $C$ is a set of added voters that makes $p$ win the election with $\card{C}\leq k$. Since $\alpha_4>0$, we know that $p$ must gain one point against every candidate $z\in Z$, and can gain against one of these candidates for each added vote, at least $\card{Z}=k$ votes must be added, so $\card{C}=k$. We prove that $C$ is a cover. By the above, $C$ covers $Z$. Indirectly assume that there is some $x$ ($y$) that is voted in the first (second) place in more than one of the voters from $C$. In this case, $x$ ($y$) gains at least $2
\alpha_1+(k-2)
\alpha_4$ ($2
\alpha_2+(k-2)
\alpha_4$) points, and so ends up with $-\alpha_1-(k-1)
\alpha_4+2
\alpha_1+(k-2)
\alpha_4=\alpha_1-\alpha_4$ points ($\alpha_2-\alpha_4$ points). Since $\alpha_1\ge\alpha_2>\alpha_4$, $p$ does not win the election in this case, a contradiction. This completes the proof.
\end{proof}

\subsection{Proof of Theorem~\ref{theorem:summary finitely many coefficients}.\ref{theorem part:alpha1 neq alpha2 alpha1 neq 2 alpha2}}

\summaryfinatelymanycoefficients*

The proof of this theorem uses two different reductions depending on which of the cases $\alpha_1>2\alpha_2$ or $\alpha_1<2\alpha$ applies. We start with the case $\alpha_1>2\alpha_2$:

\begin{theorem}\label{theorem:alpha1 > 2alpha2}
 $f$-CCAV is NP-complete for the generator $f=(\alpha_1,\alpha_2,\dots,\alpha_2,0)$ with $\alpha_1>2\alpha_2>0$.
\end{theorem}

\begin{proof}
 We again reduce from 3DM. So let $k=\card{X}=\card{Y}=\card{Z}$. We introduce the following candidates:
 \begin{itemize}
  \item the preferred candidate $p$,
  \item a candidate $c$ for each $c\in X\cup Y\cup Z$,
  \item for each $S_i\in M$, a candidate $S_i$,
  \item a dummy candidate $d$ (used for the application of Lemma~\ref{lemma:coefficients realization}).
 \end{itemize}
 
 Without loss of generality, due to Proposition~\ref{p:differ}, we can assume that for the relevant length, $\gcd{\alpha_1,\alpha_2}=1$, and $\alpha_1,\alpha_2\in\naturals$. In particular, the number $1$ can be obtained as a linear combination of $\alpha_1$ and $\alpha_2$. We thus can apply Lemma~\ref{lemma:coefficients realization} to construct a set of voters $V$ such that the relative scores are as follows:
 
 \begin{itemize}
  \item $\score{p}=0$,
  \item $\score{y}=\alpha_2$ for each $y\in Y$,
  \item $\score{x}=\score{z}=\alpha_2-\alpha_1$ for each $x\in X$ and $z\in Z$,
  \item $\score{S_i}=2\alpha_2-\alpha_1+1$ for each $S_i\in M$,
  \item $\score{d}<-3k
\alpha_1$.
 \end{itemize}
 
 For each $S_i=(x_i,y_i,z_i)\in M$, we add the following available voters:
 
 \begin{itemize}
  \item $S_i>\dots>y_i$,
  \item $z_i>\dots>S_i$,
  \item $x_i>\dots>S_i$.
 \end{itemize}
 
 We say that a vote \emph{vetoes} the candidate it places in the last position. We claim that $p$ can be made a winner of the election with adding at most $3k$ votes if and only if the 3DM-instance is positive. By construction, the dummy candidate $d$ cannot win the election by adding at most $3k$ votes, so  we ignore $d$ in the sequel.
 
 Note that since $p$ is voted in one of the middle positions in all of the potentially added voters, the score of $p$ only depends on the number of added voters, and not on the concrete choice of added votes. We thus consider only the change of relative scores between the candidates and $p$ when studying the effect of added voters.
 
 First assume that there is a cover $C$ with $\card C=k$. For all $S_i\in C$, we add all votes voting $S_i$ in the first or in the last position. Then the scores (relatively to $p$) change as follows:
 
 \begin{itemize}
  \item each $y_i$ is vetoed once, and so loses $\alpha_2$ points relative to $p$, so has final score $0$,
  \item each $S_i$ with $i\in I$ appears once in the first position and twice in the last position, so $S_i$ gains $\alpha_1-\alpha_2-2\alpha_2=\alpha_1-3\alpha_2$ points. So the final score of such an $S_i$ is $2\alpha_2-\alpha_1+1+\alpha_1-3\alpha_2=-\alpha_2+1\leq 0$, since $\alpha_2\ge1$.
  \item for each $S_i$ with $i\notin I$, the relative points of $S_i$ and $p$ do not change. Since $\alpha_1>2\alpha_2$ and $\alpha_1,\alpha_2\in\naturals$, it follows that $2\alpha_2-\alpha_1+1\leq 0$ and so $S_i$ does not beat $p$.
  \item each $z_i$ and each $x_i$ gain exactly $\alpha_1-\alpha_2$ points against $p$, and so tie with $p$.
 \end{itemize}
 
 Thus $p$ wins the election after adding these voters.
 
 For the converse, assume that there is a set $V$ of available voters that we can add with $\card V\leq 3k$ such that $p$ wins the resulting election. Since each $y\in Y$ is currently winning against $p$, and the only way for $y$ to lose points relatively to $p$ is adding the vote voting $y$ last, we know that for each $y\in Y$, there is one vote in $V$ that votes $y$ last. In particular, we have $\card{V}\ge k$. So there are votes $v_1,\dots,v_k\in V$, each voting a different $y$ in the last position, and thus each voting a different $S_i$ in the first position.
 
 For each such vote, one of the $S_i$ gains $(\alpha_1-\alpha_2)$ points against $p$. If such an $S_i$ is vetoed only once, then its final score is
 $
 2\alpha_2-\alpha_1+1
 +\alpha_1-\alpha_2-\alpha_2
 =1
 $,
 and so $p$ does not win the election. Thus each $S_i$ gaining points in one of the votes $v_1,\dots,v_k$ must be vetoed at least twice. Since $\card{V}\leq 3k$, we know that $\card{V}=3k$. Let $I=\set{i\ \vert\ S_i>\dots>y_i\in V}$. We claim that $I$ is a cover. By the above, $I$ covers each $y\in Y$. Now assume that for $i\neq j\in I$, we have $x_i=x_j$. Then $x_i$ gains $2(\alpha_1-\alpha_2)$ points against $p$, and so $p$ does not win the election, since $\alpha_1>\alpha_2$. The same argument holds for $y$. Thus, due to cardinality, we have a cover of $X$, $Y$, and $Z$.
\end{proof}

The case $\alpha_1<2\alpha_2$ is similar:

\begin{theorem}\label{theorem part:alpha1 < 2alpha2}
 $f$-CCAV is NP-complete for the generator $f=(\alpha_1,\alpha_2,\dots,\alpha_2,0)$ with $\alpha_1<2\alpha_2$.
\end{theorem}

\begin{proof}
 Very similar to the proof of Theorem~\ref{theorem:alpha1 > 2alpha2}. We again reduce from 3DM. We use the same candidate set as in the proof of Theorem~\ref{theorem:alpha1 > 2alpha2} (including the dummy candidate for the application of Lemma~\ref{lemma:coefficients realization}), and set up the scores of the nondummy candidates as follows:
 
 \begin{itemize}
  \item $\score{p}=0$,
  \item $\score{y}=\alpha_2-\alpha_1$ for each $y\in Y$,
  \item $\score{x}=\score{z}=\alpha_2$ for each $x\in X$ and each $z\in Z$,
  \item $\score{S_i}=\mathop{min}(0,3\alpha_2-2\alpha_1)$.
 \end{itemize}
 
 For each tuple $S_i=(x_i,y_i,z_i)$, we add three available voters
 
 \begin{itemize}
  \item $y_i>\dots>S_i$,
  \item $S_i>\dots>x_i$,
  \item $S_i>\dots>z_i$.
 \end{itemize}
 
 We say that a vote as above \emph{approves} (\emph{vetoes}) the candidate put in its first (last) position.

 If there is a cover with size at most $k$, we choose the $3$ voters associated with each element of the cover. This lets each $x$ and each $z$ lose $\alpha_2$ points against $p$, each $y_i$ gains $(\alpha_1-\alpha_2)$ points against $p$, and each $S_i$ in the cover ends up with 
 $
 \score{S_i}+2(\alpha_1-\alpha_2)-\alpha_2=
 \score{S_i}+2\alpha_1-3\alpha_2\leq 
 3\alpha_2-2\alpha_1+2\alpha_1-3\alpha_2=0
 $
 points (all points counted relative to $p$). Thus all candidates tie and $p$ is a winner of the resulting election.
 
 For the converse, assume that $p$ can be made a winner by adding at most $3k$ voters, let $V$ be the corresponding set. Clearly, each $x\in X$ and each $z\in Z$ need to be vetoed by some vote in $V$. Thus there are at least $2k$ votes in $V$ that vote some $S_i$ in the first position. Let $C$ contain all indices $S_i$ such that there is a vote in $V$ approving of $S_i$. Clearly, $\card{C}\ge k$.
 Then, each $S_i\in C$ gains at least $(\alpha_1-\alpha_2)$ points from these votes, and thus has $\mathop{min}(0,3\alpha_2-2\alpha_1)+\alpha_1-\alpha_2$ points. If the minimum is $0$, then this clearly is more than $0$, if the minimum is $3\alpha_2-2\alpha_1$, then the score adds up to
 $3\alpha_2-2\alpha_1+\alpha_1-\alpha_2
 =
 2\alpha_2-1\alpha_1$, 
 which also exceeds $0$ since $\alpha_1<2\alpha_2$. Thus each $S_i\in C$ must be vetoed at least once.
 In particular, there must be at least $k$ votes vetoing some $S_i$, and since due to cardinality reasons, there can be only $k$ votes vetoing some $S_i$, we know that $\card{C}\leq k$. Due to the above, we also know $\card{C}\ge k$, so  $\card{C}=k$.
 
 It remains to show that $C$ covers $X$, $Y$, and $Z$. As argued above, each $x$ and each $z$ need to be vetoed once, so $C$ covers $X$ and $Z$. To show that $C$ also covers $Y$, it suffices to show that no $y$ can be approved twice. This trivially follows since $\score{y}=\alpha_2-\alpha_1$, and each approval lets $y$ gain $\alpha_1-\alpha_2>0$ points against $p$.
\end{proof}

The proof of Theorem~\ref{theorem:summary finitely many coefficients}.\ref{theorem part:alpha1 neq alpha2 alpha1 neq 2 alpha2} now directly follows from the above results:

\begin{proof}
 If $\alpha_1\geq\alpha_2>0$ with $\alpha_1\notin\set{\alpha_2,2\alpha_2}$, then one of the following cases applies:
 \begin{itemize}
  \item $\alpha_1>2\alpha_2>0$, in this case the result follows from Theorem~\ref{theorem:alpha1 > 2alpha2},
  \item $2\alpha_2>\alpha_1>\alpha_2>0$, in this case the result follows from Theorem~\ref{theorem part:alpha1 < 2alpha2}.
    \qedhere
 \end{itemize}
\end{proof}

\subsection{Proof of Theorem~\ref{theorem:summary finitely many coefficients}.\ref{theorem part:alpha1>alpha2>alpha5}}

\summaryfinatelymanycoefficients*

\begin{proof}
 We again reduce from 3DM, let $k=\card{X}=\card{Y}=\card{Z}$.  Using Lemma~\ref{lemma:coefficients realization}, we set up the the relative scores as follows:
 
 \begin{itemize}
  \item $\score{p}=0$,
  \item $\score{x}=-(\alpha_1-\alpha_2)$ for each $x\in X$,
  \item $\score{y}=\alpha_2-\alpha_5$ for each $y\in Y$,
  \item $\score{z}=\alpha_2$ for each $z\in Z$,
  \item $\score{d}<-k
\alpha_1$, here $d$ again is a dummy candidate required for the application of Lemma~\ref{lemma:coefficients realization}, who cannot win the election and whom we ignore in the sequel.
 \end{itemize}
 
 The available voters are as follows: For each $(x,y,z)$ in $M$, there is an available voter voting 
 
 $$x>\dots>y>z.$$
 
 We claim that $p$ can be made a winner of the election if and only if the 3DM-instance is positive. First assume that the instance is positive, we then add the $k$ votes corresponding to the cover. Then:
 
 \begin{itemize}
  \item each $x\in X$ gains exactly $\alpha_1-\alpha_2$ points relatively to $p$,
  \item each $y\in Y$ loses exactly $\alpha_2-\alpha_5$ points relatively to $p$,
  \item each $z\in Z$ loses exactly $\alpha_2$ points relatively to $p$.
 \end{itemize}
 
 Thus all candidates tie and $p$ is a winner of the election. For the converse, assume that $p$ can be made a winner by adding at most $k$ voters, let $C$ be the corresponding subset of $M$. Since each $y\in Y$ and each $z\in Z$ need to lose points relatively to $p$, $C$ covers each $Y$ and $Z$. $C$ cannot cover any $x\in X$ twice, since otherwise $x$ would gain $2(\alpha_1-\alpha_2)$ points and thus beat $p$ in the election. Thus $C$ covers each $x$ at most once, and so covers each $x$ exactly once, thus $C$ is a cover as required.
\end{proof}

\subsection{Proof of Theorem~\ref{theorem:summary finitely many coefficients}.\ref{theorem part:alpha1=alpha2>alpha5>0}}

\summaryfinatelymanycoefficients*

\begin{proof}
 We again reduce from 3DM. Let $M$ be a 3DM-instance with $\card{M}=n$, and let $k=\card{X}=\card{Y}=\card{Z}$. In addition to the preferred candidate $p$, we introduce a candidate $c$ for each $c\in X\cup Y\cup Z$, and for each $S_i\in M$, we introduce a candidate $S_i$ and a candidate $S_i'$. We use Lemma~\ref{lemma:coefficients realization} to set up the registered voters such that, relatively to $p$, we have the following scores:
 
 \begin{itemize}
  \item $\score{p}=0$,
  \item $\score{x}=\score{y}=\score{z}=\alpha_1$ for each $x\in X$, $y\in Y$, and $z\in Z$,
  \item $\score{S_i}=\mathop{min}(\alpha_1,2(\alpha_1-\alpha_5))$ for each $S_i\in M$,
  \item $\score{S_i'}=\alpha_1-\alpha_5$ for each $S_i\in M$,
  \item $\score{d}<-(n+2k)
\alpha_1$, where $d$ again is a dummy candidate needed to apply Lemma~\ref{lemma:coefficients realization}, whom we ignore from now on.
 \end{itemize}
 
 Note that since $\alpha_1>\alpha_5>0$, all of these scores---except for $p$ and $d$---are strictly positive. For each $S_i=(x_i,y_i,z_i)\in M$, we add the following available voters:
 
 \begin{itemize}
  \item $\dots>S_i'>S_i$,
  \item $\dots>S_i>x_i$,
  \item $\dots>S_i>y_i$,
  \item $\dots>S_i'>z_i$.
 \end{itemize}

 Again, we say that a vote \emph{vetoes} its last-places candidate. We claim that $p$ can be made a winner by adding at most $n+2k$ voters if and only if the 3DM-instance is positive. First assume that there is a cover $C\subseteq M$ with $\card{C}\leq k$. We add the following voters:
 
 \begin{itemize}
  \item for each $S_i=(x_i,y_i,z_i)\in C$, we add the voters $\dots>S_i>x_i$, $\dots>S_i>y_i$, and $\dots S_i'>z_i$.
  \item for each $S_i\notin C$, we add the voter $\dots>S_i'>S_i$.
 \end{itemize}
 
 When adding these voters, each candidate $c\in X\cup Y\cup Z$ loses $\alpha_1$ points against $p$ and thus ends with $0$ points. Each candidate $S_i'$ loses $\alpha_1-\alpha_5$ points and so also has $0$ points in the end. A candidate $S_i\in C$ loses $2
(\alpha_1-\alpha_5)\ge\score{S_i}$ points, and a candidate $S_i\notin C$ loses $\alpha_1\ge\score{S_i}$ points. In both cases, $S_i$ has at most $0$ points after the election. Thus $p$ wins the election.
 
 For the converse, assume that there is a set $V$ of voters with $\card{V}\leq n+2k$ such that $p$ wins the election after the votes in $V$ are added. Without loss of generality, we can assume $\card V=n+2k$, since none of the available votes hurt $p$. Let $I$ be the set of indices $i$ such that a vote of the form $\dots>S_i>x_i$, $\dots>S_i>y_i$, or $\dots>S_i'>z_i$ is contained in $V$. We observe the following:
 
 \begin{itemize}
  \item Since each candidate from $X$, $Y$, and $Z$ must lose $\alpha_1$ points against $p$, for each of these candidates there must be a vote in $V$ voting that candidate last. Thus at least $3k$ votes voting one such candidate last must appear in $V$.
  \item Without loss of generality, we can assume that exactly $3k$ votes of the above form appear in $V$, since it does not help to veto one of these candidates twice (to stop one of the candidates $S_i$, $S_i'$ from winning against $p$, a vote of the form $\dots>S_i'>S_i$ is always preferable).
  \item Thus there are exactly $n-k$ votes of the form $\dots >S_i'>S_i$ in $V$. Since one such vote is enough to lower the relative score of $S_i'$ and $S_i$ to the score of $p$ and below, we can without loss of generality assume that these votes are all distinct.
  \item Let $C$ be the set of indices such that $\dots >S_i'>S_i\notin V$. Due to the above, we know that $I=k$.
  \item For each $i\in C$, since $S_i$ must lose more than $\alpha_1-\alpha_5$ points, both voters $\dots>S_i>x_i$ and $\dots>S_i>y_i$ must appear in $V$.
  \item For each $i\in C$, since $S_i'$ must lose points, the votes $\dots>S_i'>z_i$ must appear in $V$,
  \item Since there are only $3k$ voters voting a candidate from $X\cup Y\cup Z$ last, it follows that $C$ is a cover.
    \qedhere
 \end{itemize}
\end{proof}

\fi

\iffalse

\input{obsolete.tex}

\fi

\ifshownotesappendix

\section{Notes}

\begin{itemize}
 \item Is the question ``is there a subset of these votes whose addition lets $p$ be the winner'' always easy? There might be votes to add that do not give $p$ more points than everyone else, but which we still need.
\end{itemize}

\fi

\end{document}